\documentclass[noinfoline,stslayout]{imsart}

\usepackage{amsfonts,amssymb,amsmath,amscd,amsthm,latexsym}
\usepackage{natbib}
\usepackage{graphicx}
\usepackage[]{units}
\usepackage{algorithmicx}
\usepackage{algpseudocode}
\usepackage{qtree}
\usepackage{algorithm}
\usepackage{xcolor}
\usepackage{comment}
\usepackage{morefloats}
\allowdisplaybreaks

\usepackage{amsthm}

\newtheorem{lemma}{Lemma}[section]
\newtheorem{theorem}{Theorem}[section]
\newtheorem{exemp}{Example}[section]

\newcommand\oh{\nicefrac{1}{2}}
\newcommand\dd{\mathrm{d}}

\usepackage{amsthm}
\usepackage{bigints}
\usepackage{amsmath}
\usepackage{natbib}
\usepackage[colorlinks,citecolor=blue,urlcolor=blue,filecolor=blue,backref=page]{hyperref}
\usepackage{graphicx}
\usepackage{amsfonts,amssymb,amsmath,amscd,amsthm,latexsym}
\usepackage{natbib}
\usepackage{graphicx}
\usepackage[]{units}

\usepackage{algorithmicx}
\usepackage{algpseudocode}
\usepackage{qtree}
\usepackage{algorithm}
\usepackage{xcolor}
\usepackage{comment}

\startlocaldefs
\endlocaldefs

\begin{document}


\begin{frontmatter}
\title{Jeffreys priors for mixture estimation: properties and alternatives}

\author{ Clara Grazian\thanks
{Corresponding Author: Nuffield Department of Medicine, University of Oxford, John Radcliffe Hospital, Microbiology Department, Headley Way, Oxford, OX3 9DU, United Kingdom. 
mail: clara.grazian@ndm.ox.ac.uk
} \and Christian P. Robert\thanks{CEREMADE Universit\'e Paris-Dauphine, University of Warwick and CREST, Paris.
e-mail: xian@ceremade.dauphine.fr. }
}

\begin{abstract}
While Jeffreys priors usually are well-defined for the parameters of mixtures of distributions, they are not available
in closed form. Furthermore, they often are improper priors. Hence, they have never been used to draw inference on the
mixture parameters.The implementation and the properties of Jeffreys priors in several mixture settings are studied.  It is shown that the associated posterior distributions most often are improper. Nevertheless, the Jeffreys prior for the mixture weights conditionally on the parameters of the mixture components will be shown to have the property of conservativeness with respect to the number of components, in case of overfitted mixture and it can be therefore used as a default priors in this context. 
\end{abstract}


\end{frontmatter}

\section{Introduction}
\label{intro}

Bayesian inference in mixtures of distributions has been studied quite extensively in the literature. See, e.g.,
\cite{maclachlan:peel:2000} and \cite{fruhwirth:2006} for book-long references and
\cite{lee:marin:mengersen:robert:2008} for one among many surveys. From a Bayesian perspective, one of the several
difficulties with this type of distribution,

\begin{equation}\label{eq:theMix}
\sum_{\ell=1}^k p_{\ell}\,f_\ell(x|\theta_\ell)\,,\quad \sum_{\ell=1}^k p_\ell=1\,,
\end{equation}
is that its ill-defined nature (non-identifiability, multimodality, unbounded likelihood, etc.) 
leads to restrictive prior modelling since most
improper priors are not acceptable. This is due in particular to the feature that
a sample from \eqref{eq:theMix} may contain no subset from one of the
$k$ components $f(\cdot|\theta_\ell)$ (see. e.g.,
\citealp{titterington:smith:makov:1985}). Albeit the probability of such an event is
decreasing quickly to zero as the sample size grows, it nonetheless prevents the use
of independent improper priors, unless such events are prohibited \citep{diebolt:robert:1994}. 
Similarly, the exchangeable nature of the components often induces both multimodality in the posterior distribution and
convergence difficulties as exemplified by the {\em label switching} phenomenon that is now quite well-documented
\citep{celeux:hurn:robert:2000, stephens:2000b, jasra:holmes:stephens:2005, fruhwirth:2006, geweke:2007,
puolamaki:kaski:2009}. This feature is characterized by a lack of symmetry in the outcome of a Monte Carlo Markov chain
(MCMC) algorithm, in that the posterior density is exchangeable in the components of the mixture but the MCMC sample
does not exhibit this symmetry. In addition, most MCMC samplers do not concentrate around a single mode of the posterior
density, partly exploring several modes, which makes the construction of Bayes estimators of the components much harder.

When specifying a prior over the parameters of \eqref{eq:theMix}, it is
therefore quite delicate to produce a manageable and sensible non-informative
version and some have argued against using non-informative priors
in this setting (for example, \cite{maclachlan:peel:2000} argues that it is
impossible to obtain proper posterior distributions from fully noninformative
priors), on the basis that mixture models are ill-defined objects that
require informative priors to give a meaning to the notion of a component of
\eqref{eq:theMix}. For instance, the distance between two components needs to be
bounded from below to avoid repeating the same component indefinitely.
Alternatively, the components all need to be informed by the data, as
exemplified in \cite{diebolt:robert:1994} who imposed a completion scheme
(i.e., a joint model on both parameters and latent variables) such that all
components were allocated at least two observations, thereby ensuring that the
(truncated) posterior was well-defined. \cite{wasserman:2000} proved ten years
later that this truncation led to consistent estimators and moreover that only
this type of priors could produce consistency. While the constraint on the
allocations is not fully compatible with the i.i.d. representation of a mixture
model, it naturally expresses a modelling requirement that all components have
a meaning in terms of the data, namely that all components genuinely
contributed to generating a part of the data. This translates as a form of weak
prior information on how much one trusts the model and how meaningful each
component is on its own (by opposition with the possibility of adding
meaningless artificial extra-components with almost zero weights or almost
identical parameters).

While we do not seek Jeffreys priors as the ultimate prior modelling for non-informative settings, being altogether convinced of the lack of unique reference priors \citep{robert:2001,robert:chopin:rousseau:2009}, we think it is nonetheless worthwhile to study the performances of those priors in the setting of mixtures in order to determine if indeed they can provide a version of reference priors and if they are at least well-defined in such settings. We will show that only in very specific situations the Jeffreys prior provides reasonable inference.

In Section \ref{sec:jeffreys} we provide a formal characterisation of properness of the posterior distribution for the parameters of a mixture model, in particular with Gaussian components, when a Jeffreys prior is used for them. In Section \ref{sec:prosper} we will analyze the properness of the Jeffreys prior and of the related posterior distribution: only when the weights of the components (which are defined in a compact space) are the only unknown parameters it turns out that the Jeffreys prior (and so the relative posterior) is proper; on the other hand, when the other parameters are unknown, the Jeffreys prior will be proved to be improper and in only one situation it provides a proper posterior distribution. In Section \ref{sec:alternative} we present a way to realize a noninformative analysis of mixture models, in particular we propose to use the Jeffreys prior as a default prior in case of overfitted mixtures and introduce improper priors for at least some parameters. The default proposal of Section \ref{sec:alternative} will be tested on several simulation studies in Section \ref{sec:simu} and several real examples in Section \ref{sec:examples}, on both well known datasets in the mixture literature and a new dataset. Section \ref{sec:concl} concludes the paper. 

\section{Jeffreys priors for mixture models}
\label{sec:jeffreys}

We recall that the Jeffreys prior was introduced by \cite{jeffreys:1939} as a default prior based
on the Fisher information matrix
\begin{equation}
\label{eq:jeffprior}
\pi^\text{J}(\theta) \propto |I(\theta)|^{\oh} = \left\lvert -\mathbb{E} \left[ \frac{\partial^2}{\partial\theta\partial\theta^\text{T}} \log g(X; \theta)\right]\right\rvert^{\oh}  \,,
\end{equation}
whenever the later is well-defined; $I(\cdot)$ stands for the expected Fisher information matrix and the symbol $|\cdot|$ denotes the determinant. Although the prior is endowed with some frequentist properties like matching and asymptotic
minimal information \citep[Chapter 3]{robert:2001}, it does not constitute the ultimate answer to the selection of prior
distributions in non-informative settings and there exist many alternatives such as reference priors
\citep{berger:bernardo:sun:2009}, maximum entropy priors \citep{rissanen:2012}, matching priors
\citep{ghosh:carlin:srivastava:1995}, and other proposals \citep{kass:wasserman:1996}. In most settings Jeffreys priors
are improper, which may explain for their conspicuous absence in the domain of mixture estimation, since the latter
prohibits the use of independent improper priors by allowing any subset of components to go ``empty" with positive probability.
That is, the likelihood of a mixture model can always be decomposed as a sum over all possible partitions of the data
into $k$ groups at most, where $k$ is the number of components of the mixture. This means that there are terms in this
sum where no observation from the sample brings any amount of information about the parameters of a specific component. 

Approximations of the Jeffreys prior in the setting of mixtures can be found,
e.g., in \cite{figueiredo:jain:2002}, where the authors revert to independent
Jeffreys priors on the components of the mixture. This induces the same
negative side-effect as with other independent priors, namely an impossibility
to handle improper priors. \cite{rubio:steel:2014} provides a closed-form
expression for the Jeffreys prior for a location-scale mixture with two
components. The family of distributions considered in \cite{rubio:steel:2014}  is
$$
\dfrac{2\epsilon}{\sigma_1}f\left(\frac{x-\mu}{\sigma_1}\right)\mathbb{I}_{x<\mu}+
\dfrac{2(1-\epsilon)}{\sigma_2}f\left(\frac{x-\mu}{\sigma_2}\right) \mathbb{I}_{x>\mu}
$$
(which thus hardly qualifies as a mixture, due to the orthogonality in the supports of both components that allows to identify which component each observation is issued from). The factor $2$ in the fraction
is due to the assumption of symmetry 
around zero for the density $f$. For this specific model, if we impose that the weight $\epsilon$ is a function of the variance parameters,
$
\epsilon=\nicefrac{\sigma_1}{\sigma_1+\sigma_2},
$
the Jeffreys prior is given by
$$
\pi(\mu,\sigma_1,\sigma_2) \propto \nicefrac{1}{\sigma_1\sigma_2\{\sigma_1+\sigma_2\}}.
$$
However, in this setting, \cite{rubio:steel:2014} demonstrates that the posterior associated with the (regular)
Jeffreys prior is improper, hence not relevant for conducting inference. 
\cite{rubio:steel:2014} also considers alternatives to the genuine Jeffreys prior, either by reducing the range or even
the number of parameters, or by building a product of conditional priors. They further consider so-called non-objective priors that are only relevant to the specific
case of the above mixture.

Another obvious explanation for the absence of Jeffreys priors is computational, namely the closed-form
derivation of the Fisher information matrix is analytically unavailable.
The reason is that the generic $[j,h]$-th element, with $j,h \in \{ 1,\cdots,k\}$, of the Fisher information matrix for mixture models is an integral of the form
\begin{equation}
\label{eq:Fisher-elem}
-\bigintss_{\mathcal{X}} \frac{\partial^2 \log \left[\sum\limits_{\ell=1}^k p_{\ell}\,f_\ell(x|\theta_\ell)\right]}{\partial \theta_j \partial \theta_h}\left[\sum_{\ell=1}^k p_\ell\,f_{\ell}(x|\theta_\ell)\right]^{-1} d x
\end{equation}
(in the special case of component densities with a univariate parameter) which cannot be computed analytically.
Since these are unidimensional integrals, we derive an approximation of the
elements of the Fisher information matrix based on Riemann sums. The resulting
computational expense is of order $\mathrm{O}(b^2)$ if $b$ is the total number
of (independent) parameters.  Since the elements of the information matrix
usually are ratios between the component densities and the mixture density,
there may be difficulties with non-probabilistic methods of integration. 

\section{Characterization of the Jeffreys priors for mixture models and respective posteriors}
\label{sec:prosper}

Unsurprisingly, most Jeffreys priors associated with mixture models are improper, the exception being when only the weights of the mixture are unknown, as already demonstrated in \cite{bernardo:giron:1988}. 

We will characterize properness and improperness of Jeffreys priors and derived posteriors, when some or all of the parameters of distributions from location-scale families are unknown. These results are analytically established; the behavior of the Jeffreys prior and of the deriving posterior has also been studied through simulations, with sufficiently large Monte Carlo experiments (see Section \ref{sec:simu}).
The following results are often presented for Gaussian mixture models, anyway, the Jeffreys prior has a behavior common to all the location-scale families; therefore the results may be generalized to any location-scale family. 

\subsection{Weights of mixture unknown}
\label{sub:weights}

A representation of the Jeffreys prior and the derived posterior distribution for the weights of a three-component mixture
model is given in Figure \ref{weights-priorpost}: the prior distribution is much more concentrated around extreme
values in the support, i.e., it is a prior distribution conservative in the number of important components. 

\begin{figure}
\centering
\includegraphics[width=12cm, height=7.5cm]{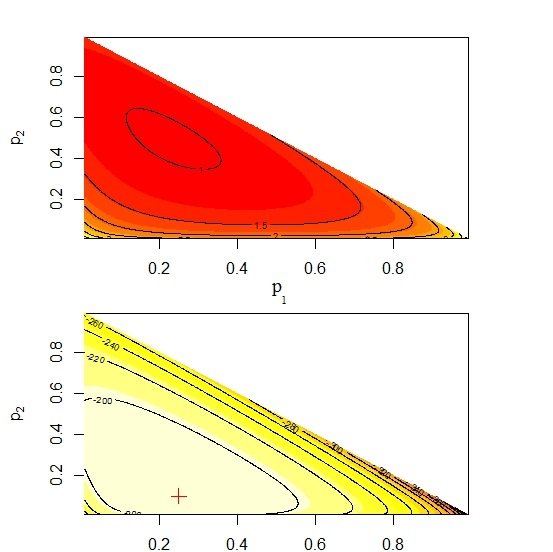}
\caption{Approximations (on a grid of values) of the Jeffreys prior (on the log-scale) when only the weights of a
Gaussian mixture model with 3 components are unknown (on the top) and of the derived posterior distribution (with known
means equal to -1, 0 and 2 respectively and known standard devitations equal to 1, 5 and 0.5 respectively). The red
cross represents the true values.} \label{weights-priorpost}
\end{figure}

\begin{lemma} 
\label{lem:weights}
When the weights $p_i$ are the only unknown parameters in \eqref{eq:theMix}, the corresponding Jeffreys prior 
is proper. 
\end{lemma}

\begin{proof}
The generic element of the Fisher information matrix $I(p)$ of the mixture model \eqref{eq:theMix} when the weights are the only unknown parameters is (for $j,h=\{1,\ldots,k-1\}$)
\begin{equation}
\label{eq:ww-prior}
\int_\mathcal{X} \frac{(f_j(x)-f_k(x))(f_h(x)-f_k(x))}{\sum_{\ell=1}^k p_\ell f_\ell(x)}
\dd x
\end{equation}
when we consider the parametrization in $(p_1,\ldots,p_{k-1})$, with
$$
p_k=1-p_1-\cdots-p_{k-1}\,.
$$
Consider now a data augmented model, where a latent variable describing the allocations of each observation to the particular component is introduced. In other words, a latent variable $z_i$ is considered such that $z_i=(0\cdots 1\cdots 0)$,
%
%
where $z_{i\ell} = 1$ in the $\ell$-th position of the vector if $x_i$ has been generated from the
$\ell$-th components, for $i=1,\cdots,n$ where $n$ is the sample size and
$\ell=1,\cdots,k$.  Therefore, $z=(z_1,\ldots,z_n)$ is a multinomial variable for $k$ possible
outcomes such that 
\begin{align}
\label{eq:lik-da}
g(x , z | \theta, p) & = g(x | z, \theta, p) g(z | \theta , p) \nonumber = \prod_{i=1}^n g(x_i | z_i, \theta , p) g(z_i | \theta , p )  \\
							& = \prod_{i=1} ^n \prod_{\ell=1}^k \left[ f_\ell(x_i | \theta_\ell) p_\ell \right]^{\mathbb{I}_{[z_{i, \ell]}=1}}  = \prod_{\ell=1}^k \left[ \prod_{i: z_{i,\ell}=1} f_l(x_i | \theta_\ell)\right] \left[ \prod_{\ell=1}^k p_\ell ^{n_\ell} \right] 
\end{align} 
\noindent where $\mathbb{I}_{[z_{i,\ell}=1]}$ is the indicator function that $z_{i,\ell}=1
$ and $n_\ell$ is the number of allocations to the $\ell$-th component. For an
extensive review of the techniques of data augmentation in the case of mixture
models one may refer to  \cite{fruhwirth:2006}.

Equation \eqref{eq:lik-da} shows that the likelihood function is separable for
$\theta$ and $p$ and that the second part is multinomial. Therefore, when
looking for the Jefffreys prior for the weights of a complete (data-augmented)
mixture model, the elements of the Fisher information matrix are
\begin{align*}
- \mathbb{E}\left[\frac{\partial^2}{\partial p_\ell^2} \log g(x,z|\theta,p)\right] & = - \frac{n_\ell n p_\ell}{p_\ell^2} = \frac{c}{p_\ell}\\
- \mathbb{E}\left[ \frac{\partial^2}{\partial p_\ell \partial p_j} \log g(x,z|\theta,p)\right] & = 0
\end{align*}
\noindent leading to the usual Jeffreys prior associated to the multinomial model, a Dirichlet distribution $\mathcal{D}ir(\frac{1}{2}, \cdots, \frac{1}{2})$. 

The above only applies to the artificial case when the allocations $z_i$ are known. When they are unknown, it is easy to see that the log-likelihood function becomes
\begin{equation}
\label{eq:lik-da}
\log g(x|\theta, p) = \log g(x, z | \theta, p) - \sum_{i=1}^n \sum_{\ell=1}^k \mathbb{I}_{[z_{i,\ell}=1]} \log p(z_{i,\ell}=1 | x_i , \theta, p)
\end{equation}
\noindent where the second term on the right side of the equation represents the loss of information compared to the data-augmented likelihood function. Define the expected Fisher information matrix for model \eqref{eq:lik-da} (when only the weights are unknown) as $I^{data-aug} (p,\theta)$.
%
%
Therefore, since the difference between both matrices is positive definite, this implies that
\begin{align*}
\det(I(p)) &\leq \det(I^{data-aug} (p)) \\
\left[\det(I(p))\right]^{1/2} &\leq \left[\det(I^{data-aug}(p))\right]^{1/2}  \\
\pi_J(p) &\leq \pi_J^{data-aug}(p)
\end{align*}  
This results shows that the Jeffreys prior on the weights of a mixture model
when allocations are unknown is proper since bounded by the Jeffreys prior
$\mathcal{D}ir(\frac{1}{2}, \cdots, \frac{1}{2})$ for the complete model.

As a particular case, when all the mixands converge to the same distribution,
each of the elements of the form \eqref{eq:ww-prior} tends to 
\begin{equation*}
\int_\mathcal{X} \frac{(f_j(x)-f_k(x))(f_\ell(x)-f_k(x))}{f_j(x)} \dd x
\end{equation*}
\noindent which does not depend on $p$. Therefore, in this case, the
determinant of the deriving Fisher information matrix is  constant in
$p=(p_1,\cdots, p_k)$ and the resulting Jeffreys prior is uniform on the
$k$-dimensional simplex.

\end{proof}

We note that this result is a generalization to a $k$-component mixture of the
prior derived in \cite{bernardo:giron:1988} for $k=2$ (however, these authors
derive the reference prior for the limiting cases when
all the components have pairwise disjoint supports and when all the components converge to the same distribution). 
This reasoning led \cite{bernardo:giron:1988} to conclude that the usual
$\mathcal{D}(\lambda_1,\ldots,\lambda_k)$ Dirichlet prior with $\lambda_\ell \in
[\nicefrac{1}{2},1]$ for $\forall \ell=1,\cdots,k$ seems to be a reasonable
approximation. They also prove that the Jeffreys prior for the weights $p_\ell$ is
convex, with an argument based on the sign of the second derivative. 

It is important to stress that, in a mixture model setting, it is usual to saturate the model when the number of components is not surely known \textit{a priori} and consider a large number of components $k$. The main difficulty in this setting is non-identifiability, in particular the rate of estimation for the satured model is much slower than the standard $1 / \sqrt{n}$.  \cite{rousseau:mengersen:2011} have studied the effect of a prior distribution on the weights of a general mixture on regularizing the posterior distribution, i.e. consistency to a single configuration of the reduced parameter space. This is achievable with a prior which allows to empty the extra-components or to merge the existing ones. In particular, \cite{rousseau:mengersen:2011} propose a Dirichlet prior distribution, with parameters $\lambda_1,\cdots, \lambda_k$ smaller than $r / 2$ (where $r$ is the dimension of $\theta_\ell$) to empty the extra-components or larger than $r / 2$ to merge the extra-components. However, the choice of $\lambda_j \, (j=1,\cdots,k)$ is quite influential for finite sample sizes. 
The configuration studied in the proof of Lemma \ref{lem:weights} is
compatible with the Dirichlet configuration of the prior proposed by \cite{rousseau:mengersen:2011}.  
This is an important property of the Jeffreys prior, since it makes
the prior conservative in the number of the components. Namely, one can asymptotically
identify the components that are artificially added to the model but have no
meaning for the data. Moreover, it offers an automatic choice, on the contrary of the Dirichlet prior where the hyper-parameters have to been chosen. 

The shape of the Jeffreys prior for the weights of a mixture model depends on
the type of the components: see Appendix A of the Supplementary
Material for a discussion.
The marginal Jeffreys prior for the weight of one component is more
concentrated around one if that component is more informative in terms of
Fisher information matrix: for example, if we consider a two-component mixture
model with a Gaussian and a Student \textit{t} component, the Jeffreys prior for the
weights will be more symmetric as the number of degrees of freedom of the
Student \textit{t} increases.  

\subsection{Weights, location and scale parameters of a mixture model unknown}

In this Section we will consider mixtures of location-scale distributions. If the components of the mixture model \eqref{eq:theMix} are distributions from a location-scale family and the location or scale parameters of the mixture components are unknown, this turns the mixture itself into a location-scale model: 
\begin{equation}
\label{eq:mix-locscale}
p_1 f_1(x|\mu,\tau)+\sum_{\ell=2}^k p_\ell f_\ell(\frac{a_\ell+ x}{b_\ell} |\mu,\tau,a_\ell,b_\ell).
\end{equation}
As a result, model \eqref{eq:theMix} may be reparametrized following
\cite{mengersen:robert:1996}, in the case of Gaussian components
\begin{equation}
\label{reparMix}
p\mathcal{N}(\mu,\tau^2)+(1-p)\mathcal{N}(\mu+\tau\delta,\tau^2\sigma^2)
\end{equation}
\noindent namely using a reference location $\mu$ and a reference scale $\tau$ (which may be, for instance, the location and scale of a specific component). Equation \eqref{reparMix} may be generalized to the case of $k$ components as
\begin{align*}
& p\mathcal{N}(\mu,\tau^2) + \sum_{\ell=1}^{k-2} (1-p) (1-q_1) \cdots (1-q_{\ell-1})q_\ell \\
& \quad  \cdot \mathcal{N}(\mu+\tau\theta_1+\cdots+\tau\cdots\sigma_{\ell-1}\theta_\ell,\tau^2\sigma_1^2\cdots\sigma_\ell^2) + \\
& + (1-p)  (1-q_1)\cdots (1-q_{k-2}) \\
& \quad \cdot \mathcal{N}(\mu+\tau\theta_1+\cdots+\tau\cdots\sigma_{k-2}\theta_{k-1},\tau^2\sigma_1^2\cdots\sigma_{k-1}^2).
\end{align*}
Since the mixture model is a location-scale model, the Jeffreys prior is 
as in the following Lemma (see also \cite[Chapter 3]{robert:2001}).

\begin{lemma}
\label{lem:prior}
When the parameters of a location-scale mixture model are unknown, the Jeffreys prior is improper, constant in $\mu$ and powered as $\tau^{-d/2}$, where $d$ is the total number of unknown parameters of the components (i.e. excluding the weights). 
\end{lemma}

An new version of the proof, never presented before, is available in Appendix B of the Supplementary Material, while the characterization of the Jeffreys prior for $\delta$ is given in Appendix C. 

We now derive analytical characterizations of the posterior distributions associated with the Jeffreys priors for mixture models. 

Consider, first, the case where only the location parameters of a mixture model are unknown. 

There is a substantial difference between the cases where $k=2$ or $k>2$.

\begin{lemma} 
When $k=2$, the posterior distribution derived from the Jeffreys prior when only the location parameters of model \eqref{eq:mix-locscale} are unknown is proper.
\label{lem:mean-post}
\end{lemma}

The complete proof of lemma \ref{lem:mean-post} is given in Appendix D of the Supplementary Material. Here it is worth noticing that the properness of the posterior distribution in the context of Lemma \ref{lem:mean-post} depends on the representation of the mixture model as a location-scale distribution, where the second component is defined with respect to a reference component: if we focus the attention on the part of the likelihood depending only on the second component, even if the prior is constant with respect to the difference between the location parameters $\delta$ as $\delta \rightarrow \pm \infty$, the likelihood depends on $\delta$ as $\exp(-\frac{n-1}{2}\delta^2)$ and therefore the behavior of the posterior distribution is convergent. 

Figure \ref{fig:mean-priorpost} shows an approximation of the Jeffreys prior for the location parameters of a two-component Gaussian mixture model on a grid of values and confirms that the prior is constant on the difference between the means and takes higher and higher values as the difference between them increases, while the posterior distribution, even if showing the classical  multimodal nature \citep{celeux:hurn:robert:2000}, seems to concentrate around the true modes. It also appears to be perfectly symmetric because the other parameters (weights and standard deviations) have been fixed as identical. 

\begin{figure}
\centering
\includegraphics[width=10cm, height=7.5cm]{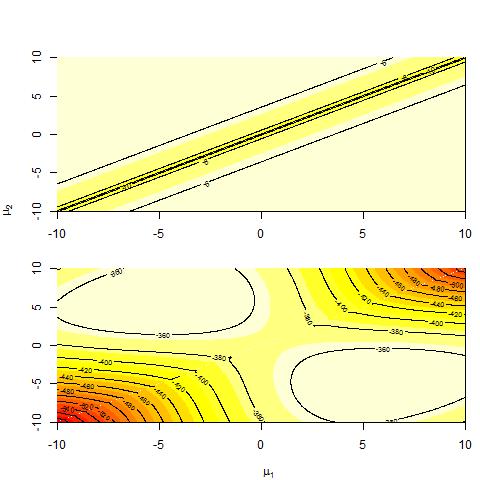}
\caption{Approximations (on a grid of values) of the Jeffreys prior (on the log-scale) when only the means of a Gaussian mixture model with two components are unknown above and of the derived posterior distribution (with known weights both equal to 0.5 and known standard deviations both equal to 5) below.}
\label{fig:mean-priorpost}
\end{figure}

The same proof cannot be extended to the general case of $k$ components, because the location parameters are defined as several distances from the reference location parameter: if we again focus the attention on the part of the likelihood depending on the second component, the integral with respect to $\delta_2$ converges, however the prior is constant with respect to any other $\delta_j$ ($j =3,\cdots,k$) as $\delta_j \rightarrow \pm \infty$ and the integral does not converge with respect to the other differences. Then the following Lemma holds (the formal proof is available in Appendix E).

\begin{lemma}
When $k>2$, the posterior distribution derived from the Jeffreys prior is improper when only the location parameters of model \eqref{eq:mix-locscale} are unknown.
\end{lemma}

This result confirms the idea that each part of the likelihood  gives
information about at most the difference between the locations of the respective
components and the reference location, but not on the locations of the other
components. 

We can now consider the case where all the parameters of \eqref{eq:mix-locscale} are unknown. 

\begin{theorem}
The posterior distribution of the parameters of a mixture model with location-scale components derived from the Jeffreys prior when all parameters of model \eqref{eq:mix-locscale} are unknown is improper.
\label{lem:all-post}
\end{theorem}

The proof is available in Appendix F of the Supplementary Material. 

\section{A noninformative alternative to Jeffreys prior}
\label{sec:alternative}

The information brought by the Jeffreys prior or lack thereof does not seem to
be enough to conduct inference in the case of mixture models. The computation
of the determinant creates a dependence between the elements of the Fisher
information matrix in the definition of the prior distribution which makes it
difficult to find and justify moderate modifications of this prior that would
lead to a proper posterior distribution. For example, using a proper prior for
part of the scale parameters and the Jeffreys prior conditionally on them does
not avoid impropriety, as it is shown Appendix G of the Supplementary Material.

The literature covers attempts to define priors that add a small amount of information that is sufficient to conduct
the statistical analysis without overwhelming the information contained in the data. Some of these are related to the
computational issues in estimating the parameters of mixture models, as in the approach of
\cite{casella:mengersen:robert:titterington:2002}, who finds a way to use perfect slice sampler by focusing on components
in the exponential family and conjugate priors. A characteristic example is given by \cite{richardson:green:1997}, who
proposes weakly informative priors, which are data-dependent (or
empirical Bayes) and are represented by flat normal priors over an interval corresponding to the range of the data.
Nevertheless, since mixture models belong to the class of ill-posed problems, the influence of a proper prior over the
resulting inference is difficult to assess.

Another solution found in \cite{mengersen:robert:1996} proceeds through the reparametrization \eqref{reparMix} and
introduces a reference component that allows for improper priors. This approach then envisions the other parameters 
as departures from the reference and ties them together by considering each parameter $\theta_\ell$ as a perturbation of
the parameter of the previous component $\theta_{\ell-1}$. This perspective is justified by the argument that the $(\ell-1)$-th
component may not be informative enough to absorb all the variability in the data. For instance, a three-component 
mixture model gets rewritten as
\begin{align*}
p\mathcal{N}(\mu,\tau^2)&+(1-p)q\mathcal{N}(\mu+\tau\theta,\tau^2\sigma_1^2) \\
						&\quad {} + (1-p)(1-q)\mathcal{N}(\mu+\tau\theta+\tau\sigma\epsilon,\tau^2\sigma_1^2\sigma_2^2)
\end{align*} 
\noindent where one can impose the constraint $1 \geq \sigma_1 \geq \sigma_2$
for identifiability reasons. Under this representation, it is possible to use
an improper prior on the global location-scale parameter $(\mu,\tau)$, while
proper priors must be applied to the remaining parameters. This
reparametrization has been used also for exponential components by
\cite{gruet:philippe:robert:1999} and Poisson components by
\cite{robert:titterington:1998}. Moreover,
\cite{roeder:wasserman:1997} proposes a Markov prior which follows the same
reasoning of dependence between the parameters for Gaussian components, where
each parameter is again a perturbation of the parameter of the previous
component $\theta_{\ell-1}$. \cite{kamary:2017} also proposes a reparametrization
of location-scale mixtures based on invariance that allows for weakly informative priors. 

On one hand, this representation suggests to define a global location-scale parameter in a more implicit way, via a hierarchical model that considers more levels in the analysis and choose noninformative priors at the last level in the hierarchy.

On the other hand, we believe that an essential feature of a default prior is that it should let the analysis be able to identify the correct number of meaningful components, in particular in the standard case where an overfitted mixture is assumed because the a priori information on the number of components is weak. 

We thus propose a prior scenario which combines both the hierarchical representation and the conservativeness property in terms of components. 

More precisely, consider the Gaussian mixture model \eqref{eq:theMix}

\begin{equation}
\label{eq:hierarc1}
g(x|\boldsymbol{\theta})=\sum_{\ell=1}^k p_i \mathcal{N}(x|\mu_\ell,\sigma_\ell).
\end{equation}

The parameters of each component may be considered as related in some way; for example, the observations induce a
reasonable range, which makes it highly improbable to face very different means in the above Gaussian mixture model. A
similar argument may be used for the standard deviations. 

Therefore, at the second level of the hierarchical model, we may write

\begin{align}
\label{eq:hierarc2}
\mu_\ell & \stackrel{iid}{\sim} \mathcal{N}(\mu_0, \zeta_0) \nonumber \\
\sigma_\ell & \stackrel{iid}{\sim} \frac{1}{2} \mathcal{U}(0,\zeta_0) + \frac{1}{2}\frac{1}{\mathcal{U}(0,\zeta_0)} \nonumber \\
p|\mu,\sigma & \sim \pi^J(p|\mu,\sigma)
\end{align}

\noindent which indicates that the location parameters vary between components,
but are likely to be close, and that the scale parameters may be smaller or
larger than $\zeta_0$; we have decided to define both $\mu_{\ell}$ and $\sigma_{\ell}$ as depending on hyperparameter $\zeta_0$ without loss of generality, as one may notice by analysing mean and variance of the random variables; this representation allows the application of the MCMC scheme proposed in \cite{robert:mengersen:1999} which allows a better mixing of the chains. 
The mixture weights are given the prior distribution $\pi^J(p|\mu,\sigma)$ which is the Jeffreys prior for the weights, conditional on the location and scale parameters, given in Section \ref{sub:weights}; this choice makes use of the conservative property of the Jeffreys prior for the weights which is essential in the case of
miss-specification of the number of components. 

At the third level of the hierarchical model, the prior may be noninformative:

\begin{align}
\label{eq:hierarc3}
\pi(\mu_0,\zeta_0) \propto \frac{1}{\zeta_0}.
\end{align}

As in \cite{mengersen:robert:1996} the parameters in the mixture model are
considered tied together; on the other hand, this feature is not obtained via a
constrained representation of the mixture model itself, but via a hierarchy in
the definition of the model and the parameters.   

\begin{theorem}
\label{theo:post-altern}
The posterior distribution derived from the hierarchical representation of the Gaussian mixture model 
associated with \eqref{eq:hierarc1}, \eqref{eq:hierarc2} and \eqref{eq:hierarc3}
is proper. 
\end{theorem}

The proof of Theorem \ref{theo:post-altern} is available in Appendix H of the Supplementary Material. 

As a side remark, even if Theorem  \ref{theo:post-altern} is stated for Gaussian mixture models, it may be extended to other location-scale distributions. Section \ref{sec:examples} will present an example with log-normal components, Section \ref{sub:network} with Gumbel components. However it cannot be generalized to any location-scale distribution. 

This hierarchical version of the mixture model presents some advantages; in particular, the Jeffreys prior used for the weights is conservative in terms of number of components in the case of misspecification. We remind that when the number of components is not known, it is usual in practice to fix a model with a high number of components (if one wants to avoid a nonparametric analysis), therefore it is essential that the posterior distribution gives hints on the right $k$. This feature of the Jeffreys prior allow the experimenter to do so in a noninformative way. More precisely, this hierchical prior respect the Assumption 5 of \cite{rousseau:mengersen:2011}.

\section{Simulation Study}
\label{sec:simu}

In this Section we present the results of several simulations studies we conduct to support the theoretical results presented so far. The results of additional simulations are given in Appendix G and H of the Supplementary Material. 

As a remark, integrals of the form \eqref{eq:Fisher-elem} need to be approximated, as mentioned in Section \ref{sec:jeffreys}. There are numerical issue here. We decided to use Riemann sums (with $550$ points) when the component standard deviations are sufficiently large, as they produce stable results, and Monte Carlo integration (with sample sizes of $1500$) when they are small. In the latter case, the variability of MCMC results seems to decrease as $\sigma_i$
approaches $0$. See the Supplementary Material for a detailed description of
these computational issues. 

We can analyse the property of conservativeness in overfitted mixtures through simulations, by using the hierarchical prior proposed in Section \ref{sec:alternative}. We consider a very simple example to illustrate this theoretical result. 
Suppose we want to fit a two-component Gaussian mixture model with weights $p$ and $1-p$ and parameters unknown to a sample of data $\mathbf{x}=\{x_1,\cdots,x_n\}$ generated from a standard normal distribution $\mathcal{N}(0,1)$. We computed the posterior distribution for $M=20$ replications of samples of size $n=(50,100,500,1000,5000,10000)$. Figure \ref{fig:hierch-overfitting} shows that the posterior means of $p$ increases to $1$ as $n$ increases. 

\begin{figure}
\centering
\includegraphics[scale=0.5]{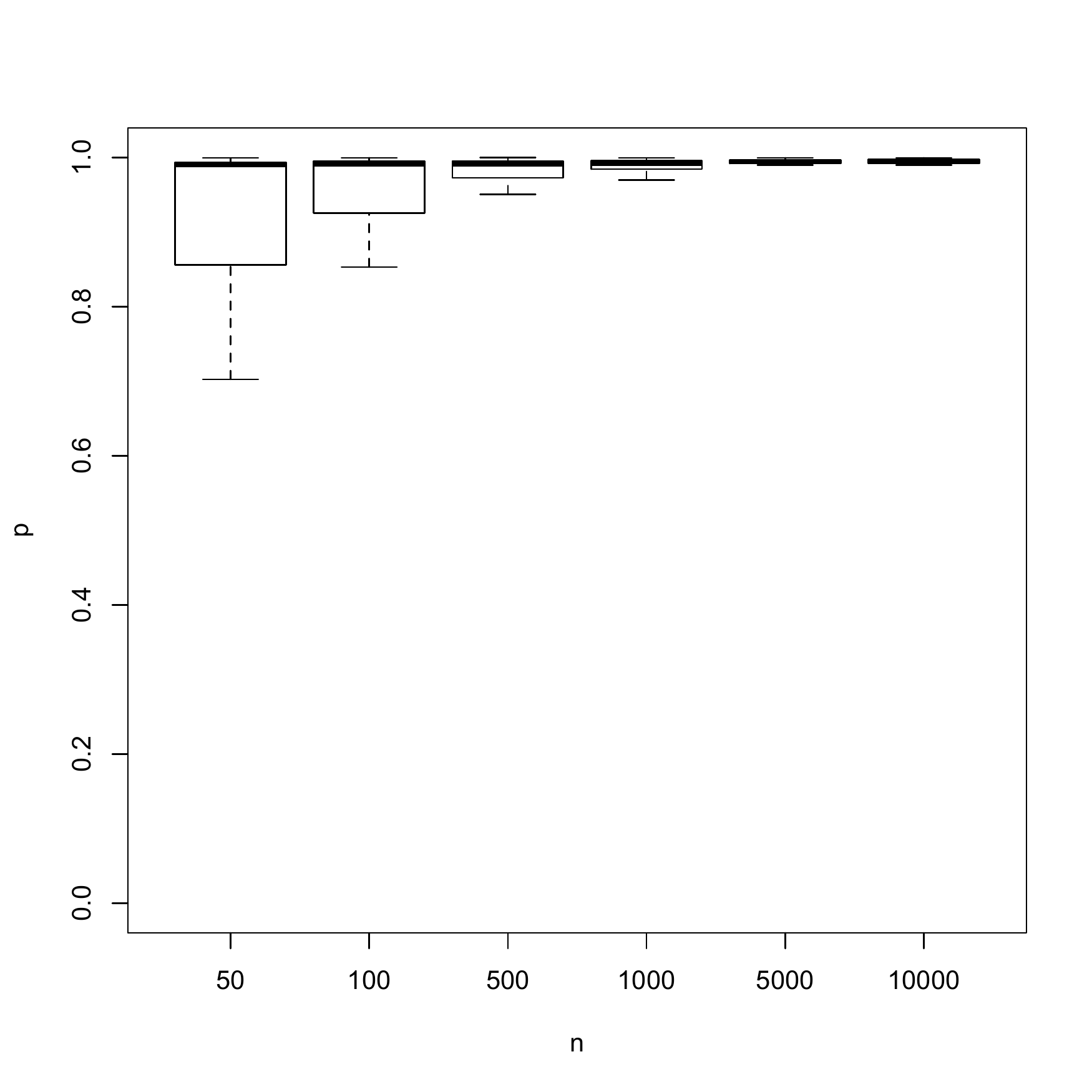}
\caption{Boxplots of posterior means of the largest weight $p$, with the hierarchical prior on the parameters, in particular a conditional Jeffreys prior on the weights, for sample sizes $n=50,100,500,1000,5000,10000$. }
\label{fig:hierch-overfitting}
\end{figure}

We have also considered a more complicated situation, where we want to fit a model with an increasing number of components ($k=(2,3,4,5)$) to a data set $\mathbf{x}=\{x_1,\cdots,x_n\}$ generated from a two-component mixture model

\begin{equation}
\label{eq:overfit}
0.5 \mathcal{N}(-3,1) + 0.5 \mathcal{N}(3,1).
\end{equation}

Figures \ref{fig:hierch-overfitting-incrK} and
\ref{fig:hierch-overfitting-incrK2} show the boxplots for the posterior means
of the weights obtained through $M=20$ replications of the experiment, with a
correct ($k=2$) or a misspecified ($k=(3,4,5)$) model. It is clear that as the
number of components increases, the additional weights are estimated by smaller
and smaller values as the sample size increases. It is evident that the variability of the estimates (in repetitions of the experiment) is smaller when an exact number of components is assumed; however, in every case, the Bayesian analysis based on the Jeffreys prior is able to identify the right number of components. The higher variability in estimating the weights is reflected in the fact that, as the
number of components increases, the estimated (and the predictive) densities are less and less
smooth, nevertheless this feature is mitigated as the sample size increases, see Appendix H in the Supplementary Material.

\begin{figure}
\centering
\includegraphics[scale=0.7]{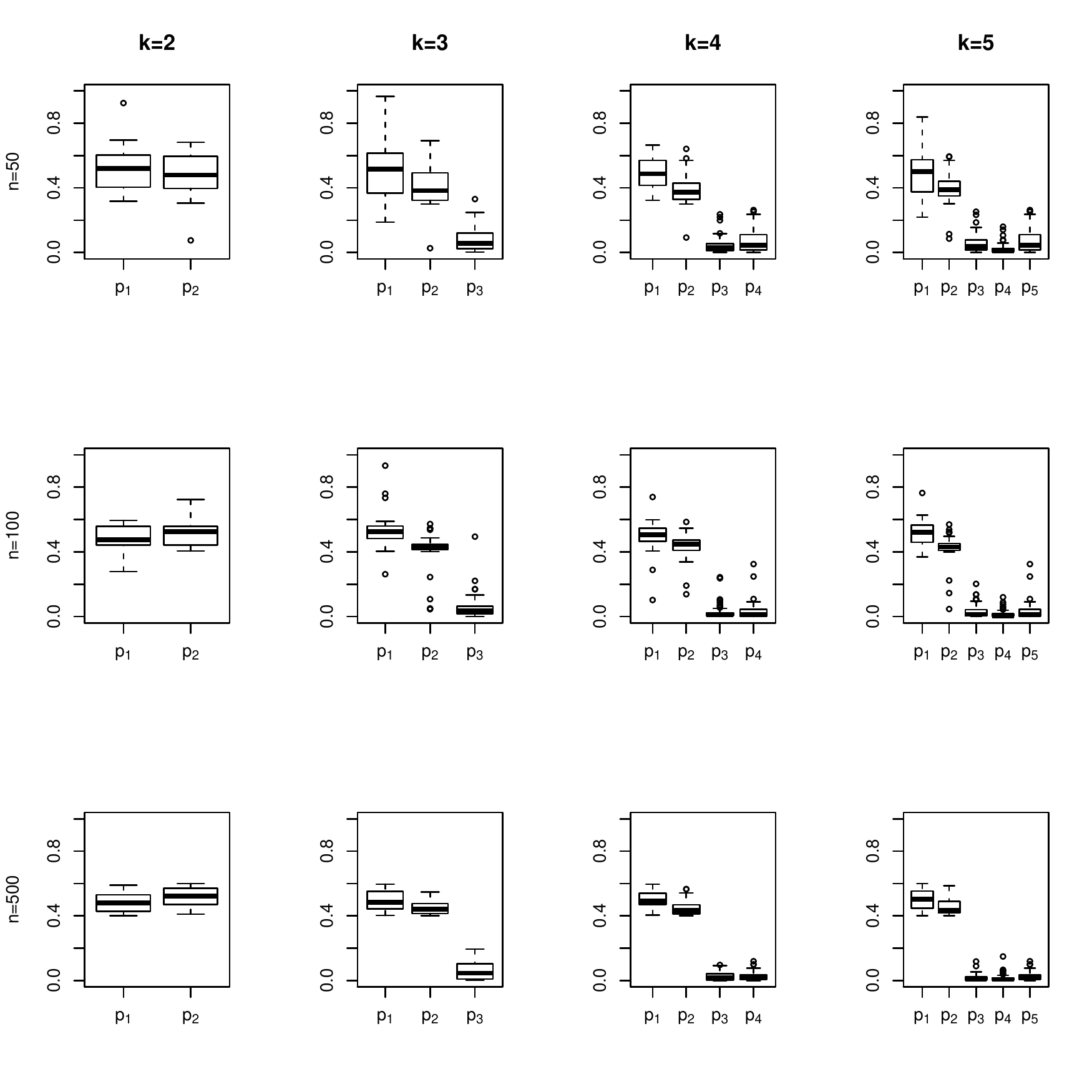}
\caption{Boxplots of posterior means of the weights $\mathbf{p}$, with the
hierarchical prior on the parameters, in particular a conditional Jeffreys
prior on the weights, for sample sizes $n=(50,100,500)$ and with models with
$k=(2,3,4,5)$ components.}
\label{fig:hierch-overfitting-incrK}
\end{figure}

\begin{figure}
\centering
\includegraphics[scale=0.7]{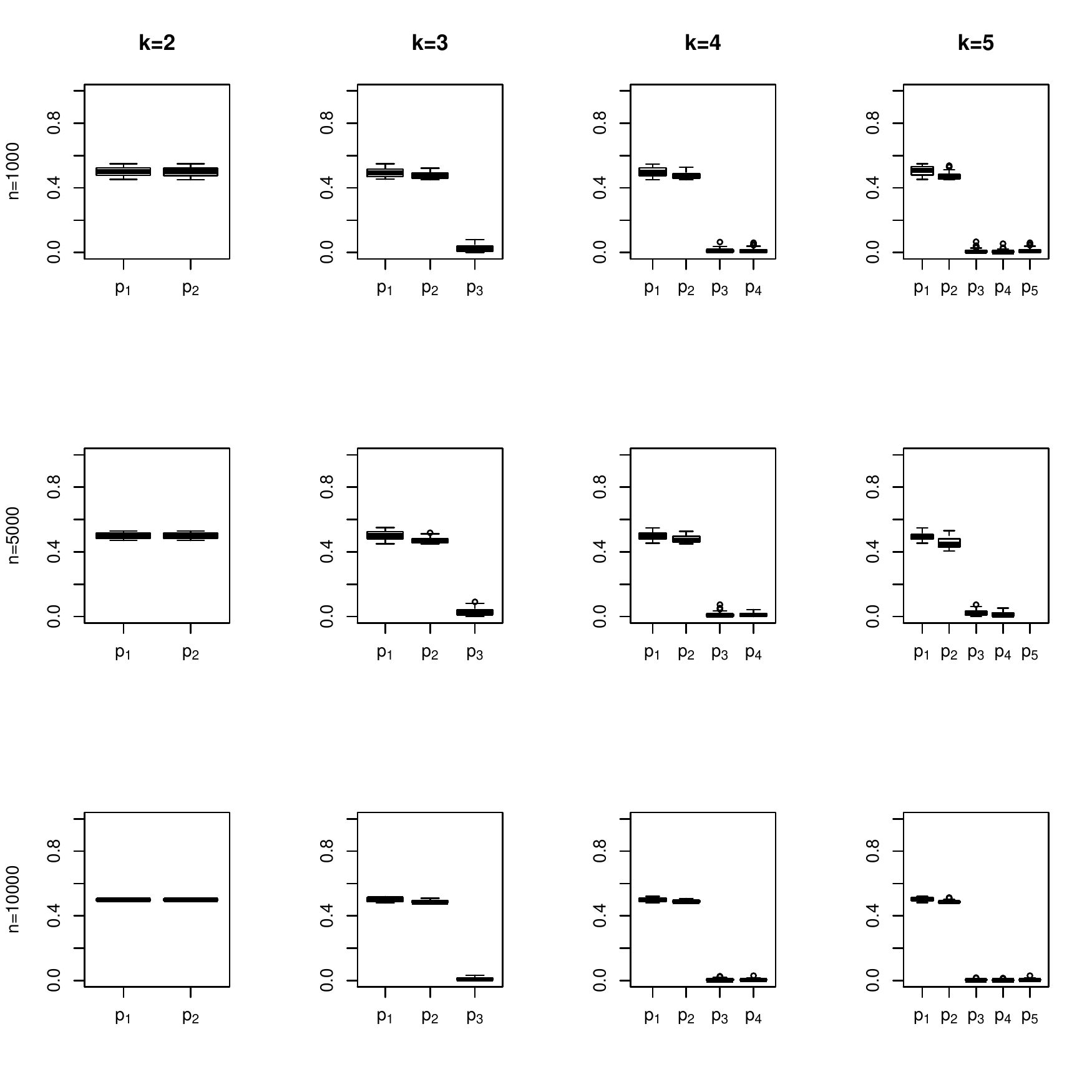}
\caption{As in Figure \ref{fig:hierch-overfitting-incrK}, for sample sizes $n=(1000,5000,10000)$.}
\label{fig:hierch-overfitting-incrK2}
\end{figure}

\section{Illustrations}
\label{sec:examples}

In this Section we will analyse the performance of the approach proposed in
Section \ref{sec:alternative} in three datasets so well-known in the literature
of mixture models that they can be taken as benchmarks and in a new dataset we propose here for the first time. In order to better present this new dataset, the analysis of it is presented separately.  

The first dataset contains data about the velocity (in km per second) of 82
galaxies in the Corona Borealis region. The goal of this analysis is to
understand the number of stellar populations, in order to support a particular
theory of the formation of the Galaxy. The Galaxy dataset has been investigated by
several authors, including \cite{richardson:green:1997}, \cite{raftery:1996},
\cite{escobar:west:1995} and \cite{roeder:1990} among others.

The galaxies velocities are considered as random variables distributed
according to a mixture of $k$ normal distributions. The evaluation of the
number of components has proved to be delicate, with estimates from 3 in
\cite{roeder:wasserman:1997} to 5 in \cite{richardson:green:1997} and 7 in
\cite{escobar:west:1995}. 

We have assumed a ten-component mixture model and check whether or not the hierarchical
approach that uses the conditional Jeffreys prior on the weights of the
mixture model manages to identify a smaller number of significant components. 
The results are available in Figure \ref{fig:galaxy} and Table
\ref{tab:datasets}. The algorithm identifies 5 components with weights larger
than zero, which is a result along the line of \cite{richardson:green:1997} and
more conservative than \cite{escobar:west:1995}, which confirms the Jeffreys
prior's feature of being conservative in the number of the components. Credible
intervals also show that the parameters of the components with marginal
posterior distributions for the weights not concentrated around zero are
estimated with lower uncertainty. 

\begin{figure}
\centering
\includegraphics[scale=0.5]{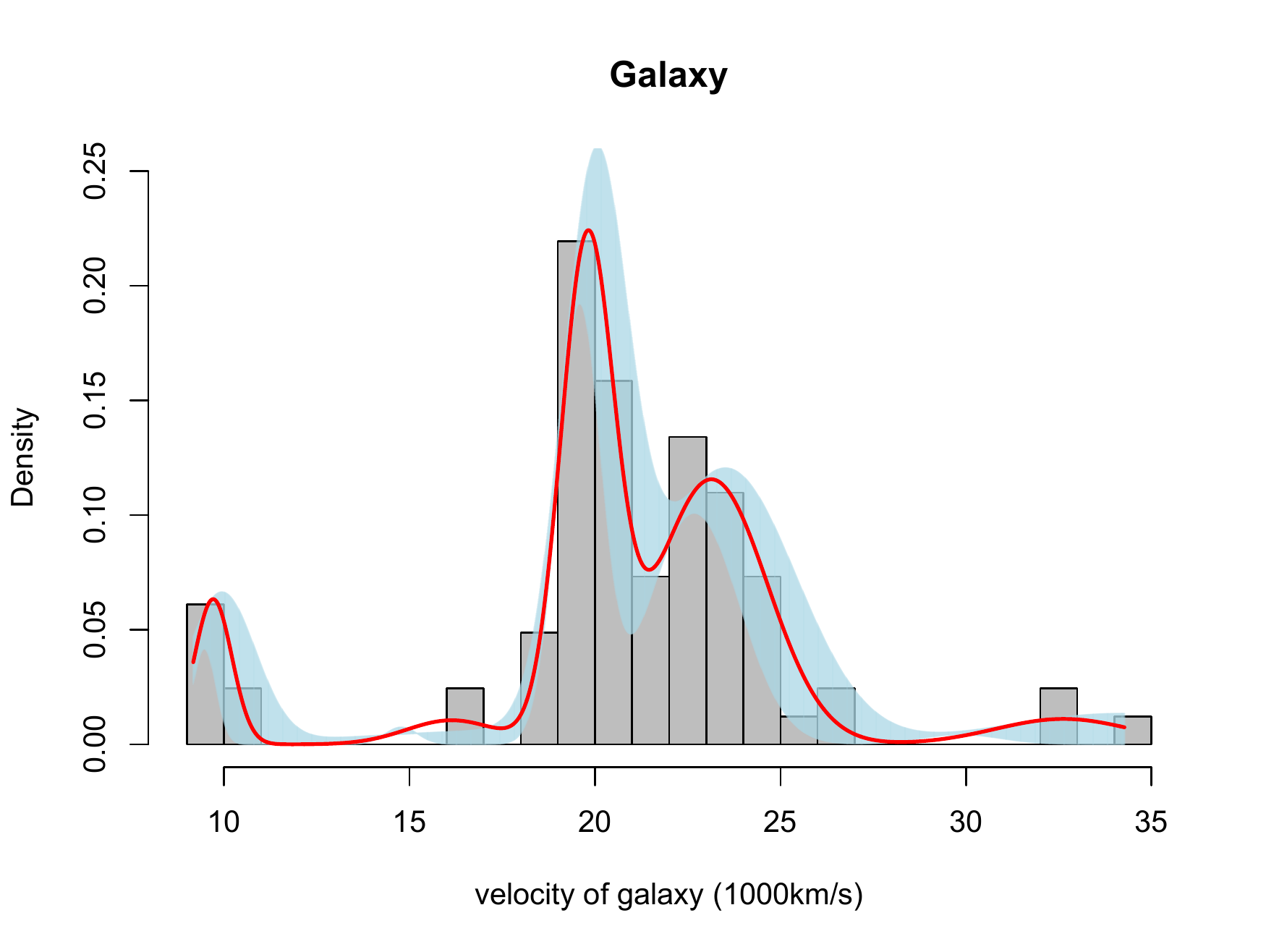}
\caption{Predictive distribution of the galaxy dataset: the red line represent the estimation of the density, the shadow blue area represents the credible intervals in $10^5$ simulations by assuming a ten-component mixture model.}
\label{fig:galaxy}
\end{figure}

The second dataset is related to a population study to validate caffeine as a
probe drug to establish the genetic status of rapid acetylators and slow
acetylators \citep{bechtel:1993}: many drugs, including caffeine, are
metabolyzed by a polymorphic enzyme (EC 2.3.1.5) in humans and the white
population is divided into two groups of slow acetylators and rapid
acetylators. Caffeine is considered an interesting drug to study the phenotype
of people, because it is regularly consumed by a large amount of the
population. Several population studies have been conducted, some of them
reporting a bimodality, some others a trimodality. We focus on the study
presented by \cite{bechtel:1993}, involving 245 unrelated patients and
computing the molar ratio between two metabolites of caffeine, AFMU and 1X,
both measured in urine 4 to 6 hours after ingestion of 200 mg of caffeine. 

We have again assumed a ten-component mixture model and checked whether or not
the hierarchical approach which uses the conditional Jeffreys prior on the
weights of the mixture model is able to identify a smaller number of
significant components. 

The results are available in Figure \ref{fig:enzyme} and Table
\ref{tab:datasets}. The algorithm identifies two components with weights clearly
larger than zero and two other components with very small weights.
\cite{bechtel:1993} identify a bimodal density, while
\cite{richardson:green:1997} consider highly likely a 3-5 component mixture. The
Jeffreys prior allows to concentrate the analysis on mainly two subgroups and
it suggests that Gaussian components may be inappropriate in this setting:
by looking to the location of the components with small weights, it may be more
adequate to consider asymmetric distributions. 

\begin{figure}
\centering
\includegraphics[scale=0.5]{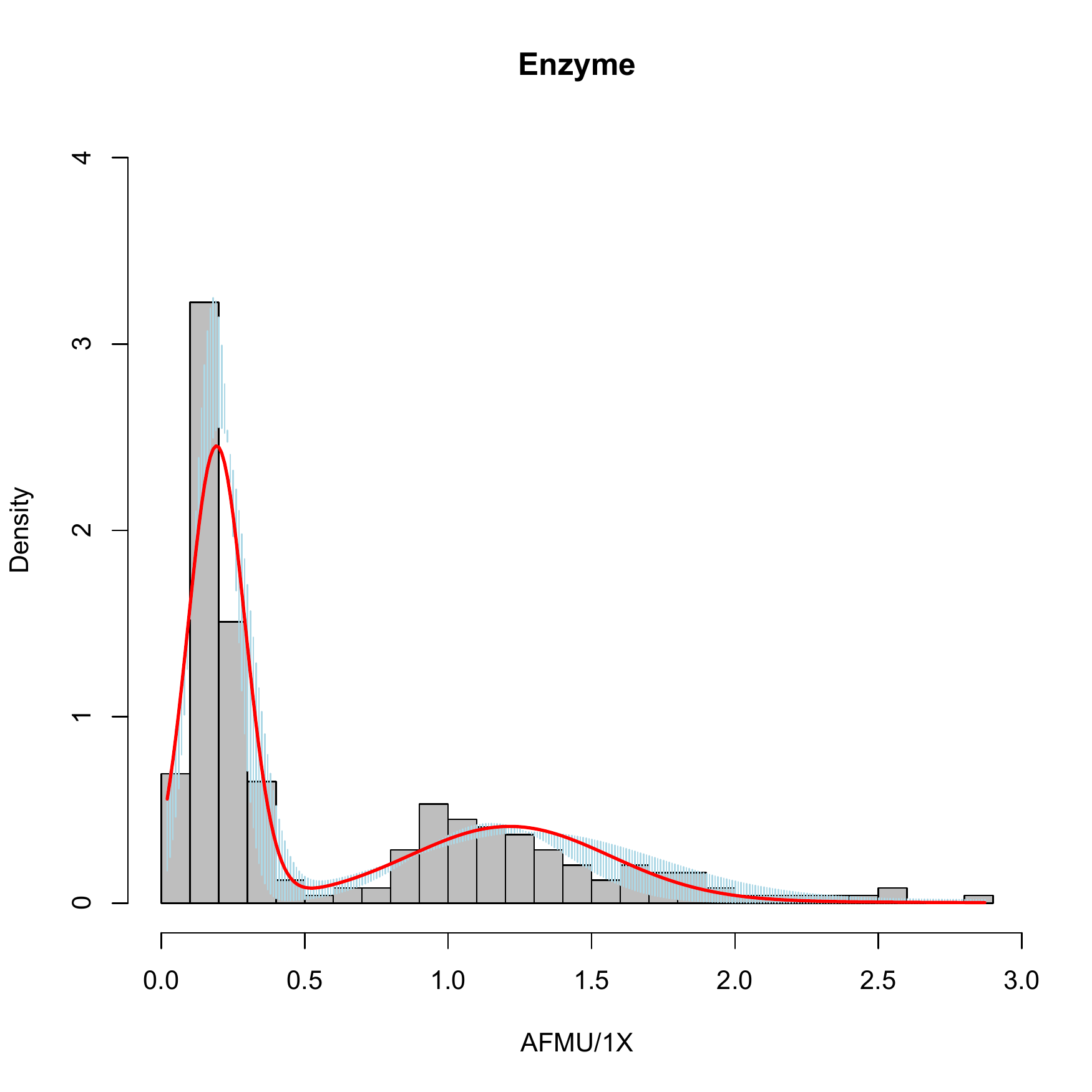}
\caption{Predictive distribution of the enzyme dataset: the red line represent the estimation of the density, the shadow blue area represents the credible intervals in $10^5$ simulations by assuming a ten-component mixture model.}
\label{fig:enzyme}
\end{figure}

Our third dataset is related to measuring the acid neutralizing capacity (ANC)
(in log-scale) of a sample of 155 lakes in north-central Wisconsin, to
determine the number of lakes that have been affected by acidic deposition
\citep{crawford:1992}: the ANC measures the capability of a lake to neutralize
acid, i.e. low values may indicate a problem for the lake's biological
diversity. 

The results are available in Figure \ref{fig:acidity} and Table
\ref{tab:datasets}. The algorithm identifies two components with significant weights and two other components with very small weights.
\cite{crawford:1992} assume a bimodal density, while
\cite{richardson:green:1997} consider highly likely a 3-5 component model. The
Jeffreys prior again allows to concentrate the analysis on two main subgroups
and suggests to investigate the importance of other two components and
possibly the goodness-of-fit of the log-normal distribution in this setting. 

\begin{figure}
\centering
\includegraphics[scale=0.5]{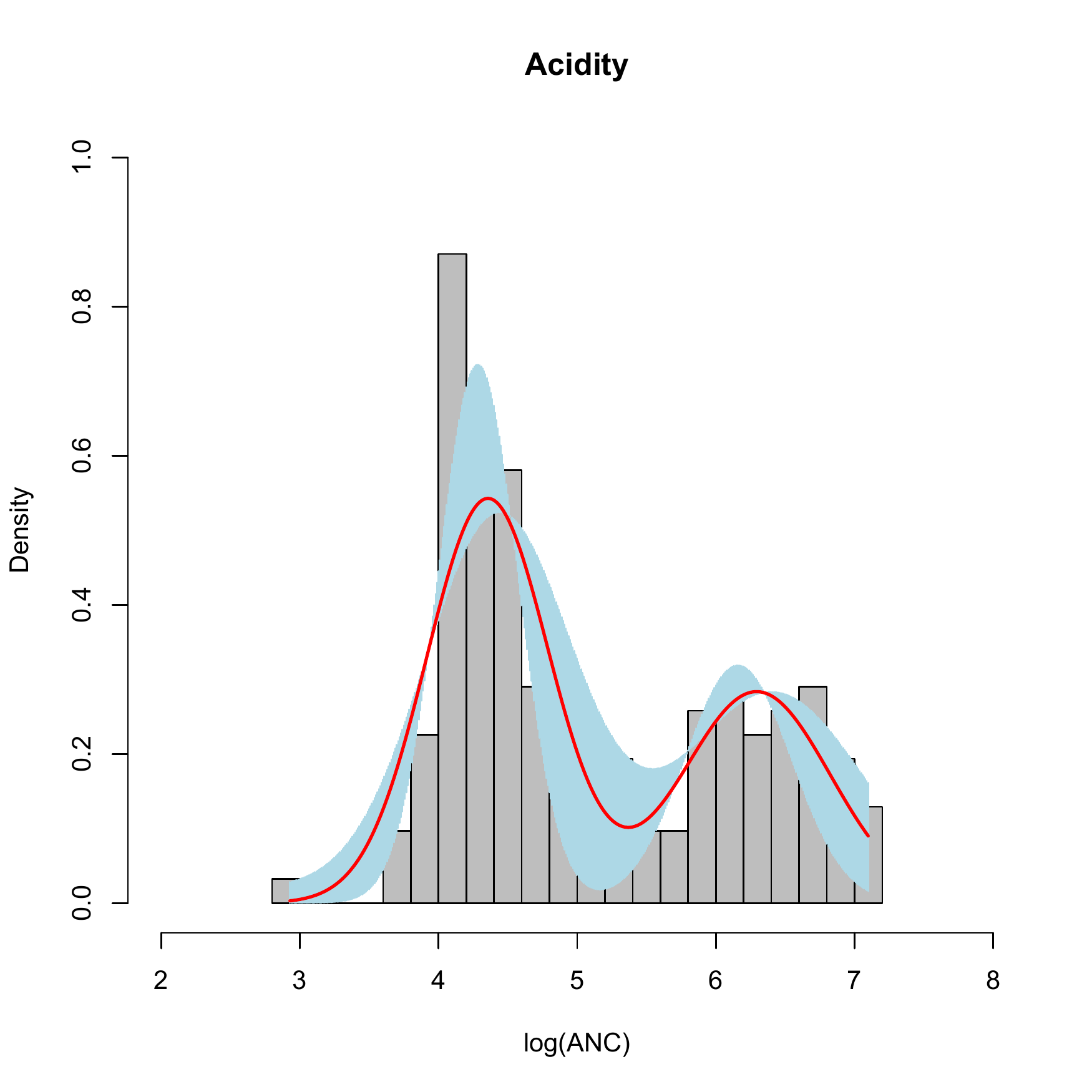}
\caption{Predictive distribution of the acidity dataset: the red line represent
the estimation of the density, the shadow blue area represents the credible
intervals in $10^5$ simulations by assuming a ten-component mixture model.}
\label{fig:acidity}
\end{figure}

\begin{table}
\centering
\caption{Posterior means for the weights, the means and the standard deviations of a ten-component mixture model, assumed for the galaxy, the enzyme and the acidity datasets (the first number in brackets is the posterior mean and the second is the posterior standard deviation). We have decided to not shown the estimated location and scale parameters when the weights are concentrated around zero.}
\label{tab:datasets}      
\begin{tabular}{llll}
\hline\noalign{\smallskip}
\textbf{Dataset:} & \textbf{galaxy} & \textbf{enzyme}  & \textbf{acidity}  \\
\noalign{\smallskip}\hline\noalign{\smallskip}
$p_1$								&	0.437							&	0.606					&	0.601				\\
{}										&	(23.139, 1.507)			&	(0.193, 0.090)		&	(4.356,0.442) 	\\
$p_2$								&	0.390							&	0.343					&	0.378				\\
{}										&	(19.790, 0.715)			& (1.216, 0.348)		&	(6.294, 0.531)	\\
$p_3$								&	0.080							& 0.021					&	0.003				\\
{}										&	( 9.709, 0.503)				& (0.915, 1.174)		& (0.083, 0.802)	\\
$p_4$								&	0.056							&	0.018					&	0.003				\\
{}										&	(32.630, 1.842)			&	(1.176, 0.702)		&	(0.125, 0.589)	\\
$p_5$								&	0.037							&	0.000					&	0.000				\\
{}										&	(16.138,1.226)				& -							& -						\\
$\sum\limits_{\ell=6}^{10} p_\ell$		&	0.000							&	0.000					& 0.000
\end{tabular}
\end{table}

\subsection{Network dataset}
\label{sub:network}

A recent trend in computer network systems is the deployment of network functions in software \cite{nunes2014survey}. The so-called ``software dataplanes'' are emerging as an alternative to traditional hardware switched and routers, reducing costs and enhancing programmability. 

The monitoring of IP packets is, among all possible network functions, one of the most suitable for a software deployment. However, the monitoring has a huge cost in terms of consumed CPU (processing) time by packet. The main reason for this is that each incoming packet triggers the retrieval, from a large hash-table, of all the information related to the packet flow (i.e. the packet's family). This operation is generally called flow-entry retrieval. The time required for the flow-entry retrieval (retrieval time) mainly depends on whether such information is available in one of the processor caches (e.g. L1, L2, L3) or in memory.    

The dataset used in this analysis consists of generated samples of retrieval time, each with $10^6$ times, under two different set-ups. In the first one, the flow-entry has been forced to reside in fast processor caches (``hit''). In the second one, all flow-entries have been forced to reside in the server RAM (memory), which results in a slower flow-entry retrieval (``miss'').

Both samples show a heavy tail, due to possible hash collisions on the table, as well as additional delays introduced by measuring the retrieval time at a nanosecond timescale. In the case of ``miss'', another reason for the heavy tail can be identified with the virtual/physical memory mapping, which can inflate the retrieval time in some cases.

The goal of a realistic analysis is to infer the proportion of reported times which may be considered from the ``hit'' distribution and the proportion of times which may be considered from the ``miss'' distribution, i.e. to derive what is the percentage of packets for which the flow-entry was in the cache and the percentage of packets for which the flow-entry was in memory.

However, a first simulation is generally used to test the procedure. The interest of the analysis will be in the region of the space where the two distributions are overlapping, therefore the interest is not in the external tails, which may, nonetheless, affect inference. Therefore, a preliminary analysis may be conducted in order to understand if a part of the future observations may be discarded from the analysis. 
In this particular case, the conservative property of the Jeffreys prior may be used in order to understand how much important are the tails of each distribution and to identify the right models to use. For instance, a comparison between a Gaussian mixture model and a mixture model with Gumbel components may be run: if in both cases the analysis run with a Jeffreys prior for the mixture weights identifies more than two (assumed) distributions of interest, this may be a suggestion that the observations allocated to the external components (not the ``hit'' or the ``miss'' ones) may be discarded, providing inference on the proportion of observations to discard as well. 

Figure \ref{fig:network} and Table \ref{tab:network} show the results of this analysis: adopting a Jeffreys prior for the mixture weights when assuming Gumbel components allows to better estimate the first component and to describe the asymmetry observed in the data as an asymmetry in the first component instead of an additional component. Nevertheless it is not sufficient to identify the observations in the right tail of the second component as part of its tail, since the algorithm identifies a third component located in that part of the space. 

In this setting, the Jeffreys prior allows to i) identify a miss-specification of the model assumptions (the approximated Bayes factor of the mixture of Gumbel components against the mixture of normal components is $2.10$) and ii) identify which part of the observations to discard from further studies. 

\begin{figure}
\centering
\includegraphics[scale=0.5]{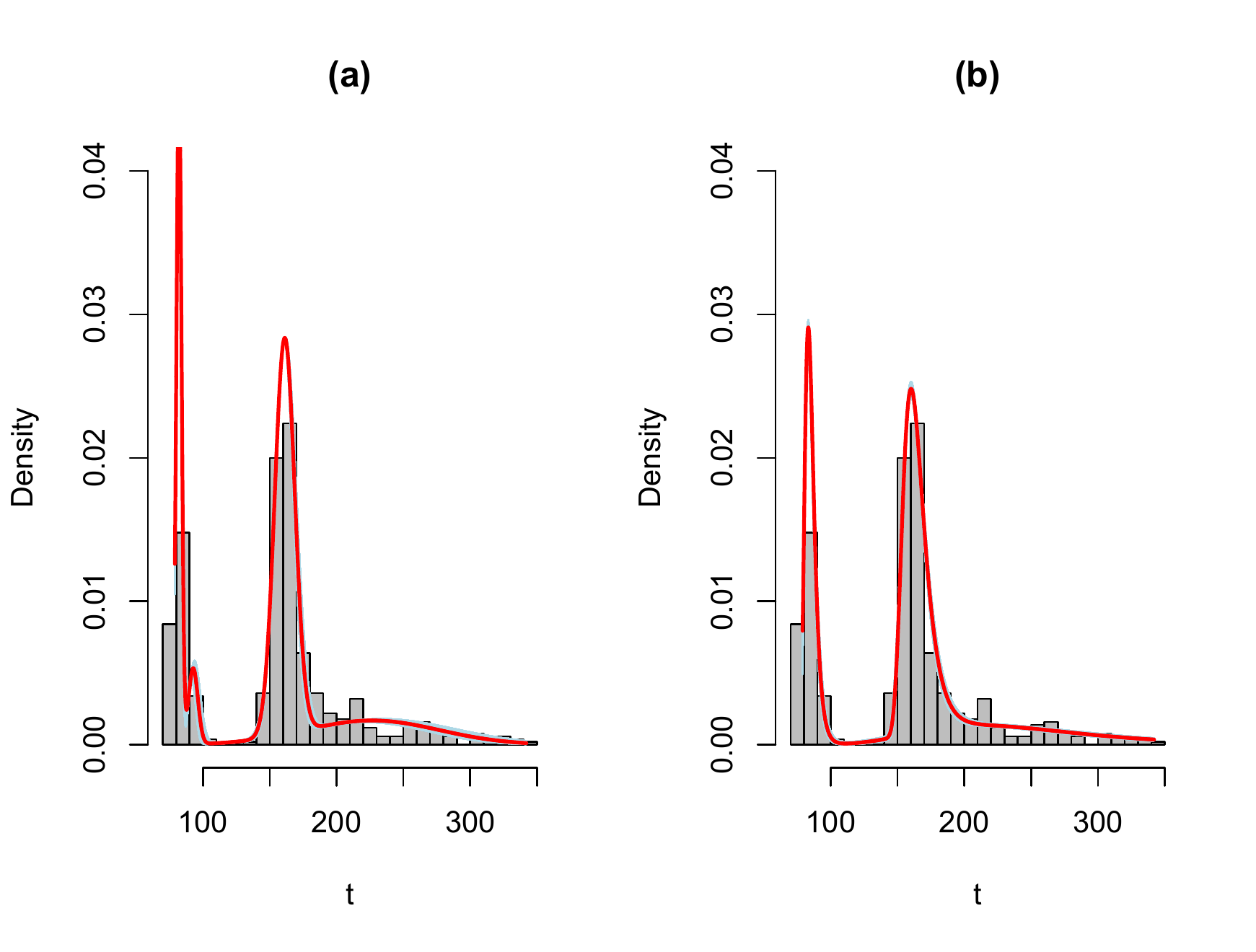}
\caption{Predictive distribution of the network dataset: the red line represent
the estimation of the density, the shadow blue area (very concentrated around the red lines) represents the credible
intervals in $10^5$ simulations by assuming a ten-component mixture model, with Gaussian components on the left and with Gumbel components on the right.}
\label{fig:network}
\end{figure}

\begin{table}
\centering
\caption{Posterior means for the weights, the means and the standard deviations of a ten-component mixture model, assumed for the network dataset (credible intervals of level 0.95 in brackets).}
\label{tab:network}      
\begin{tabular}{lll}
\hline\noalign{\smallskip}
\textbf{Gaussian comp.} & & \\
\noalign{\smallskip}\hline\noalign{\smallskip}
\textbf{$p$} & \textbf{$\mu$}  & \textbf{$\sigma$}  \\
\noalign{\smallskip}\hline\noalign{\smallskip}
0.214                                        & 224.318                                         &  50.271                                      \\
\textit{(0.180,0.249)} & \textit{(222.657,233.842)} & \textit{(45.483,55.265)} \\ 
0.519                                        & 161.645 								& 7.497                                    \\
\textit{(0.474,0.568)} & \textit{(160.216,161.882)}  & \textit{(6.830,8.212)} \\ 
0.221                                        & 82.847                                         & 1.888                                        \\
\textit{(0.188,0.257)} & \textit{(81.057,82.270)}  & \textit{(1.666,2.135)} \\ 
0.046                                        & 92.826                                         & 3.474                                      \\
\textit{(0.029,0.064)} & \textit{(91.710 ,93.700)}   & \textit{(2.698,4.388} \\ 
$\sum\limits_{\ell=5}^{10} p_\ell = 0.000$ & &\\
\hline\noalign{\smallskip}
\textbf{Gumbel comp.} & & \\
\noalign{\smallskip}\hline\noalign{\smallskip}
\textbf{$p$} & \textbf{$\mu$}  & \textbf{$\sigma$}  \\
\noalign{\smallskip}\hline\noalign{\smallskip}
0.214                                        & 213.512                                         &  59.080                                      \\
\textit{(0.183,0.251)} & \textit{(213.446,213.846)} & \textit{(53.526,64.667)} \\ 
0.520                                        & 160.164 								& 7.959                                    \\
\textit{(0.479,0.562)} & \textit{(160.113,160.482)}  & \textit{(7.465,8.482)} \\ 
0.265                                        & 83.260                                         & 3.348                                       \\
\textit{(0.219,0.302)} & \textit{(83.251,83.270)}  & \textit{(3.005,3.753)} \\ 
$\sum\limits_{\ell=4}^{10} p_\ell = 0.000$

\end{tabular}
\end{table}

\section{Conclusion}\label{sec:concl}

This thorough analysis of the Jeffreys priors in the setting of
mixtures with location-scale components shows that mixture distributions deserve the qualification of an ill-posed problem with regard to the production of non-informative priors.
Indeed, we have shown that most configurations for Bayesian inference in this
framework do not allow for the standard Jeffreys prior to be taken as a
reference.  While this is not the first occurrence where Jeffreys priors cannot
be used as reference priors, we have shown that the Jeffreys prior for the mixture weights has the important property to be conservative in the number of components, with a configuration compatible with the results of \cite{rousseau:mengersen:2011}.This is a general feature of the Jeffreys prior for the mixture weights, which is independent from the shape of the distributions composing the mixture. 

Nevertheless, we have decided to study its behavior in the specific case of components from location-scale families. We have proposed a hierarchical representation of the mixture model, which allow for improper priors at the highest level of the hierarchy and assumes the Jeffreys prior for the mixture weights in the second level, conditional on prior distributions for the location and scale parameters along the line of \cite{mengersen:robert:1996}. 

Through several examples, both on simulated and real datasets, we have shown that this representation seems to be more conservative on the number of components than other non or weakly informative prior distributions for mixture models available in the literature. In particular, it seems to be able to recognize the meaningful components, which is an essential property for a noninformative prior for mixture model: in fact, in an objective setting, it is essential to consider the possibility to have assumed a wrong number of components. 
In this sense, the Jeffreys prior for the mixture weights may be used to identify the meaningful components and possible miss-specifications of either the number or the distributional family of the components. 

As a note aside, we have mainly analyzed mixture of Gaussian distributions in this paper, with extensions of the theoretical results to the other distributions of the location-scale family. Nevertheless, the possible difficulties deriving from the use of distributions different from the Gaussian are not considered here and will be the focus of future research. In particular, all likelihoods poorly specified and ill-behaved cases are more likely to meet difficulties. However, the Jeffreys prior is known as a regularization prior that does not necessarily reflect prior beliefs, but in combination with the likelihood function yields posteriors with desirable properties; see \cite{hoogerheide:vandijk:08} for a detailed review of ill-behaved posterior cases and the role of the Jeffreys prior in those cases.

\section*{Acknowledgements and Notes}

The code used for the Gaussian mixture models is available online at the following link: \url{https://github.com/cgrazian/Jeffreys_mixtures}.

The Authors want to thank Gioacchino Tangari, from the Department of Electronic and Electrical Engineering, University College London, for having provided the simulations of Section \ref{sub:network}.

\bibliographystyle{ims}  
\hyphenation{Post-Script Sprin-ger}

\section*{Supplementary Material}

\subsection*{Appendix A: Form of the Jeffreys prior for the weights of the mixture model.}

The shape of the Jeffreys prior for the weights of a mixture model depends on
the type of the components.  Figure \ref{weights-GMM}, \ref{weights-GtMM} and
\ref{weights-GtMM-df} show the form of the Jeffreys prior for a two-component
mixture model for different choices of components. It is always concentrated
around the extreme values of the support, however the amount of concentration
around $0$ or $1$ depends on the information brought by each component. In
particular, Figure \ref{weights-GMM} shows that the prior is much more
symmetric as there is symmetry between the variances of the distribution
components, while Figure \ref{weights-GtMM} shows that the prior is much more
concentrated around 1 for the weight relative to the normal component if the
second component is a Student \textit{t} distribution. 

Finally Figure \ref{weights-GtMM-df} shows the behavior of the Jeffreys prior when the first component is Gaussian and the second is a Student \textit{t} and the number of degrees of freedom is increasing. As expected, as the Student \textit{t} is approaching a normal distribution, the Jeffreys prior becomes more and more symmetric.
\begin{figure}
\centering
\includegraphics[width=10cm, height=7.5cm]{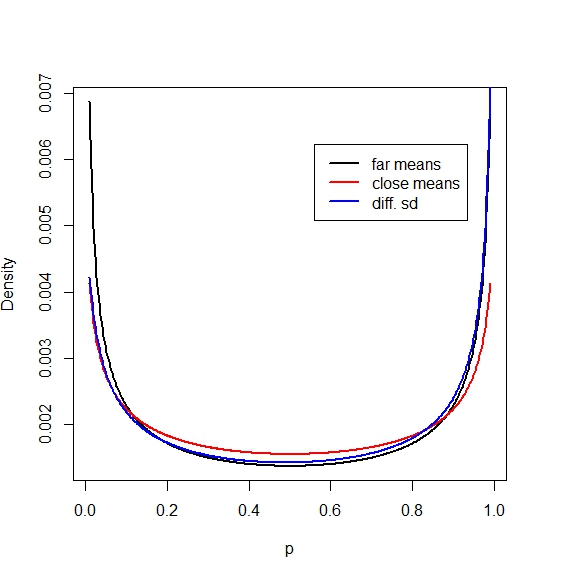}
\caption{Approximations of the marginal prior distributions for the first weight of a two-component Gaussian mixture model, $p\,\mathcal{N}(-10,1)+(1-p)\,\mathcal{N}(10,1)$ (black), $p\,\mathcal{N}(-1,1)+(1-p)\,\mathcal{N}(1,1)$ (red) and $p\,\mathcal{N}(-10,1)+(1-p)\,\mathcal{N}(10,10)$ (blue).}
\label{weights-GMM}
\end{figure}
\begin{figure}
\centering
\includegraphics[width=10cm, height=7.5cm]{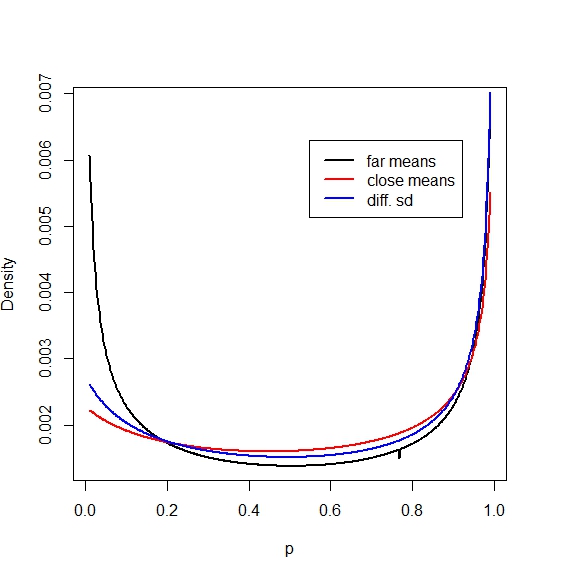}
\caption{Approximations of the marginal prior distributions for the first weight of a two-component mixture model where the first component is Gaussian and the second is Student t, $p\,\mathcal{N}(-10,1)+(1-p)\,\mathrm{t}(df=1,10,1)$ (black), $p\,\mathcal{N}(-1,1)+(1-p)\,\mathrm{t}(df=1,1,1)$ (red) and $p\,\mathcal{N}(-10,1)+(1-p)\,\mathrm{t}(df=1,10,10)$ (blue).}
\label{weights-GtMM}
\end{figure}
\begin{figure}
\centering
\includegraphics[width=10cm, height=7.5cm]{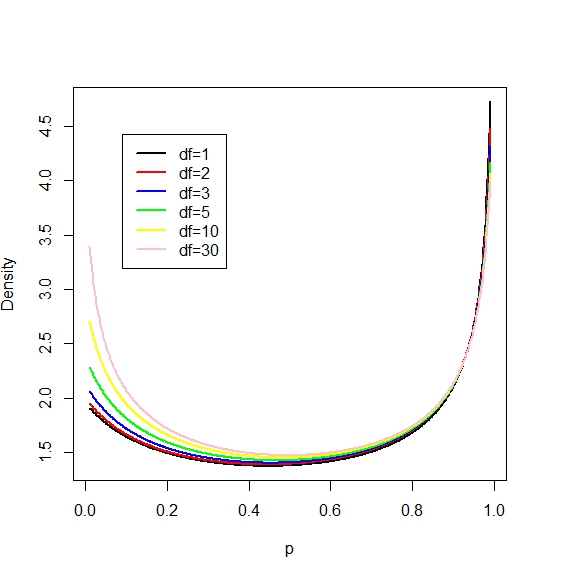}
\caption{Approximations of the marginal prior distributions for the first weight of a two-component mixture model where the first component is Gaussian and the second is Student t with an increasing number of degrees of freedom.}
\label{weights-GtMM-df}
\end{figure}

\subsection*{Appendix B: Proof of Lemma 3.2}

\textit{
When the parameters of a location-scale mixture model are unknown, the Jeffreys prior is improper, constant in $\mu$ and powered as $\tau^{-d/2}$, where $d$ is the total number of unknown parameters of the components (i.e. excluding the weights). 
}

\begin{proof} 
We first consider the case where the means are the only unknown parameters of a Gaussian mixture model
\begin{equation*}
g_X(x)=\sum_{l=1}^k p_l \mathcal{N}(x|\mu_l,\sigma_l^2)
\end{equation*}
The generic elements of the expected Fisher information matrix are, in the case of diagonal and off-diagonal terms respectively:
\footnotesize{
\begin{align*}
& \mathbb{E}\left[- \frac{\partial^2 \log g_X(X)}{\partial \mu_i^2}\right]=\frac{p_i^2}{\sigma_i^4} \bigintsss_{-\infty}^\infty
\frac{\left[ (x-\mu_i) \mathcal{N}(x|\mu_i,\sigma_i^2)\right]^2}{\sum_{\ell=1}^k p_\ell \mathcal{N}(x|\mu_\ell,\sigma_\ell^2)}
d x, \\
& \mathbb{E}\left[- \frac{\partial^2 \log g_X(X)}{\partial \mu_i \partial \mu_j}\right]=\frac{p_i p_j}{\sigma_i^2
\sigma_j^2} \cdot \\
& \quad \cdot \bigintsss_{-\infty}^\infty \frac{(x-\mu_i) \mathcal{N}(x|\mu_i,\sigma_i^2)(x-\mu_j)
\mathcal{N}(x|\mu_j,\sigma_j^2) }{\sum_{\ell=1}^k p_\ell\mathcal{N}(x|\mu_\ell,\sigma_\ell^2)}  d x.
\end{align*}
}
\normalsize{
Now, consider the change of variable $t=x-\mu_i$ in the above integrals, where $\mu_i$ is thus the mean of the $i$-th Gaussian component
($i\in\{1,\cdots,k\}$). The above integrals are then equal to
}
\footnotesize{
\begin{align*}
& \mathbb{E}\left[- \frac{\partial^2 \log g_X(X)}{\partial \mu_j^2}\right] = \frac{p_j^2}{\sigma_j^4} \cdot \\ & \quad \cdot \bigintsss_{-\infty}^\infty
\frac{\left[ (t-\mu_j+\mu_i) \mathcal{N}(t|\mu_j-\mu_i,\sigma_i^2)\right]^2}{\sum_{\ell=1}^k p_\ell
\mathcal{N}(t|\mu_\ell-\mu_i,\sigma_\ell^2)} d x,\\
& \mathbb{E}\left[- \frac{\partial^2 \log g_X(X)}{\partial \mu_j \partial \mu_m}\right] = \frac{p_j p_m}{\sigma_j^2
\sigma_m^2} \cdot \\
& \quad \cdot \bigintsss_{-\infty}^\infty \frac{(t-\mu_j+\mu_i) \mathcal{N}(x|\mu_j,\sigma_j^2)(t-\mu_m+\mu_i)
\mathcal{N}(t|\mu_m-\mu_i,\sigma_m^2) }{\sum_{\ell=1}^k p_\ell\mathcal{N}(t|\mu_\ell-\mu_i,\sigma_\ell^2)}  d x.\\
\label{eq:means-prior}
\end{align*}
}
\normalsize{
Therefore, the terms in the Fisher information only depend on the differences $\delta_j=\mu_i-\mu_j$ for $j \in
\{1,\cdots,k \}$. This implies that the Jeffreys prior is improper since a reparametrization in
($\mu_i,\mathbf{\delta}$) shows the prior does not depend on $\mu_i$.

Moreover, consider a two-component mixture model with all the parameters unknown
\begin{equation*}
p\mathcal{N}(\mu,\tau^2)+(1-p)\mathcal{N}(\mu+\tau\delta,\tau^2\sigma^2).
\end{equation*}
With some computations, it is straightforward to derive the Fisher information matrix for this model, partly shown in Table \ref{tab:FishInfo_repar}, where each element is multiplied for a term which does not depend on $\tau$.
}

\begin{table}[h]
\centering
\caption{Factors depending on $\tau$ of the Fisher information matrix for the reparametrized model (8)}.
\label{tab:FishInfo_repar}
\begin{tabular}{c|ccccc}
                  & \textbf{$\sigma$} & \textbf{$\delta$} & p                     & \textbf{$\mu$} & \textbf{$\tau$} \\ \hline
\textbf{$\sigma$} & 1                 & 1                 & \multicolumn{1}{c|}{1}           & $\tau^{-1}$      & $\tau^{-1}$     \\
\textbf{$\delta$} & 1                 & 1                 & \multicolumn{1}{c|}{1}           & $\tau^{-1}$      & $\tau^{-1}$     \\
p        & 1                 & 1                 & \multicolumn{1}{c|}{1}           & $\tau^{-1}$      & $\tau^{-1}$     \\ \cline{2-6} 
\textbf{$\mu$}    & $\tau^{-1}$       & $\tau^{-1}$       & \multicolumn{1}{c|}{$\tau^{-1}$} & $\tau^{-2}$      & $\tau^{-2}$     \\
\textbf{$\tau$}   & $\tau^{-1}$       & $\tau^{-1}$       & \multicolumn{1}{c|}{$\tau^{-1}$} & $\tau^{-2}$      & $\tau^{-2}$    
\end{tabular}
\end{table}

Therefore, the Fisher information matrix considered as a function of $\tau$ is a block matrix. From well-known results
in linear algebra, if we consider a block matrix 
\begin{equation*}
M=
\begin{bmatrix}
A & B \\
C & D
\end{bmatrix}
\end{equation*}
\noindent then its determinant is given by $\det(M)=\det(A-BD^{-1}C)\det(D)$. In the case of a two-component mixture
model where the total number of components parameters (i.e. non considering the weights) is $d=4$, $\det(D)\propto\tau^{-4}$, while $\det(A-BD^{-1}C)\propto 1$ (always interpreted as functions of $\tau$ only). Then the
Jeffreys prior for a two-component Gaussian mixture model is proportional to $\tau^{-2}$. If we generalize to the case of a Gaussian mixture model with $k$ components, the total number of component parameters is $d = 2k$ and the Jeffreys prior for a $k$-component Gaussian mixture model is proportional to $\tau^{-k}$.

When considering the general case of components from a location-scale family, this feature of improperness of the Jeffreys prior distribution is still valid, because, once reference location-scale parameters are chosen, the mixture model may be rewritten as
\begin{equation}
\label{eq:mix-locscale}
p_1 f_1(x|\mu,\tau)+\sum_{\ell=2}^k p_\ell f_\ell(\frac{a_\ell+ x}{b_\ell} |\mu,\tau,a_\ell,b_\ell).
\end{equation}
Then the second derivatives of the logarithm of model \eqref{eq:mix-locscale} behave as the ones we have derived for the Gaussian case, i.e. they will depend on the differences between each location parameter and the reference one, but not on the reference location itself. Then the Jeffreys prior will be constant with respect to the global location parameter and powered in the global scale parameter. 

\end{proof}

\subsection*{Appendix C: Jeffreys prior for $\delta=\mu_2 - \mu_1$}

\textit{The Jeffreys prior of $\delta$ conditional on $\mu$ when only the location parameters are unknown is improper.
}

\begin{proof}

When considering the reparametrization by \cite{mengersen:robert:1996}, the Jeffreys prior for $\delta$ for a fixed $\mu$ has the form:
\begin{equation*}
\pi^J(\delta|\mu)\propto \left[
\int_\mathfrak{X}\frac{\left[{(1-p)x\exp\{-\frac{x^2}{2}\}}\right]^2}{{p\sigma\exp\{-\frac{\sigma^2(x+\frac{\delta}{\sigma\tau})^2}{2}\}}+{(1-p)\exp\{-\frac{x^2}{2}\}}}
d x \right]^{\frac{1}{2}}
\end{equation*}
\noindent and the following result may be demonstrated.
The improperness of the conditional Jeffreys prior on $\delta$ depends (up to a constant) on the double integral
\begin{eqnarray*}
\int_\Delta
\int_\mathfrak{X} c \frac{\left[(1-p)x\exp\{-\frac{x^2}{2}\}\right]^2}{p\sigma\exp\{-\frac{\sigma^2(x+\frac{\delta}{\sigma\tau})^2}{2}\}+(1-p)\exp\{-\frac{x^2}{2}\}}
d x d\delta.
\end{eqnarray*}
\noindent The order of the integrals is allowed to be changed, then
\begin{eqnarray*}
\int_\mathfrak{X} x^2 \int_\Delta
\frac{\left[(1-p)\exp\{-\frac{x^2}{2}\}\right]^2}{p\sigma\exp\{-\frac{\sigma^2(x+\frac{\delta}{\sigma\tau})^2}{2}\}+(1-p)\exp\{-\frac{x^2}{2}\}}
d\delta d x .
\end{eqnarray*}
\noindent Define $f(x)=(1-p)e^{-\frac{x^2}{2}}=\frac{1}{d}$. Then
\begin{eqnarray*}
\int_\mathcal{X} x^2 \int_\Delta \frac{1}{d^2
p\sigma\exp\{-\frac{\sigma^2(x+\frac{\delta}{\sigma\tau})^2}{2}\}+d} d\delta d x .
\end{eqnarray*}
\noindent Since the behavior of $\left[d^2
p\sigma\exp\{-\frac{\sigma^2(x+\frac{\delta}{\sigma\tau})^2}{2}\}+d\right]$ depends on $\exp\{-\delta^2\}$
as $\delta$ goes to $\infty$, we have that 
\begin{equation*}
\int_{-\infty} ^{+\infty} \frac{1}{\exp\{-\delta^2\}+d} d\delta > \int_{A} ^{+\infty} \frac{1}{\exp\{-\delta^2\}+d} d\delta  
\end{equation*}
\noindent because the integrand function is positive. Then
\begin{equation*}
\int_{A} ^{+\infty} \frac{1}{\exp\{-\delta^2\}+d} d\delta > \int_{A} ^{+\infty} \frac{1}{\varepsilon+d} d\delta = +\infty.
\end{equation*}

Therefore the conditional Jeffreys prior on $\delta$ is improper.

\end{proof}

Figure \ref{fig:priorpost-diff} compares the
behaviour of the prior and the resulting posterior distribution for the difference between the means of a two-component
Gaussian mixture model: the prior distribution is symmetric and it has different behaviours depending on the value of the
other parameters, but it always stabilizes for large enough values; the posterior distribution appears to always concentrate
around the true value. 

\begin{figure}
\centering
\includegraphics[width=6.5cm, height=7.5cm]{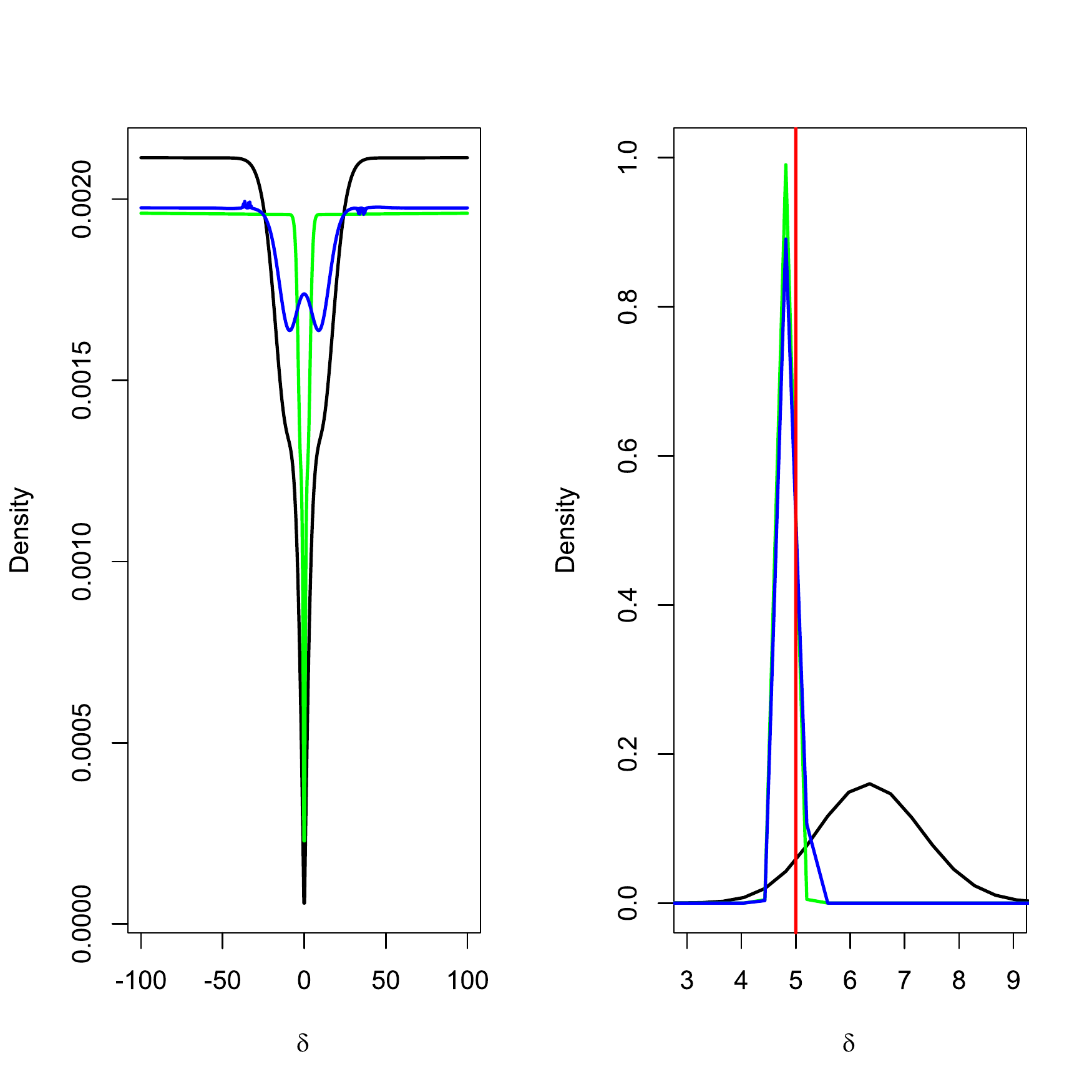}
\caption{Approximations (on a grid of values) of the Jeffreys prior (on the natural scale) of the difference between the means of a Gaussian mixture model with only the means unknown (left) and of the derived posterior distribution (on the right, the red line represents the true value), with known weights equal to $(0.5,0.5)$ (black lines), $(0.25,0.75)$ (green and blue lines) and known standard deviations equal to $(5,5)$ (black lines), $(1,1)$ (green lines) and $(7,1)$ (blue lines).} 
\label{fig:priorpost-diff}
\end{figure} 

\subsection*{Appendix D: Proof of lemma 3.3}

\textit{When $k=2$, the posterior distribution derived from the Jeffreys prior when only the location parameters of model 
\begin{equation}
\label{eq:mix-locscale}
p_1 f_1(x|\mu,\tau)+\sum_{\ell=2}^k p_\ell f_\ell(\frac{a_\ell+ x}{b_\ell} |\mu,\tau,a_\ell,b_\ell).
\end{equation}
are unknown is proper.}

\begin{proof}
The conditional Jeffreys prior for the means of a Gaussian mixture model follow the behavior of the product of the diagonal elements of the Fisher information matrix:
\begin{align*}
  \frac{p_1 p_2}{\sigma_1^2 \sigma_2^2} & \left\{ \int_{-\infty}^{+\infty}\frac{\left[
t\mathcal{N}(0,\sigma_1)\right]^2}{p_1\mathcal{N}(0,\sigma_1)+p_2\mathcal{N}(\delta,\sigma_2)} d t  \times  \int_{-\infty}^{+\infty} \frac{\left[
u\mathcal{N}(0,\sigma_2)\right]^2}{p_1\mathcal{N}(-\delta,\sigma_1)+p_2\mathcal{N}(0,\sigma_2)} d u \right. \nonumber \\ 
					&{} \left. -\left(\int_{-\infty}^{+\infty} \frac{ t\mathcal{N}(0,\sigma_1)
(t-\delta)\mathcal{N}(\delta,\sigma_2)}{p_1\mathcal{N}(0,\sigma_1)+p_2\mathcal{N}(\delta,\sigma_2)} d t\right)^2 \right\}^\frac{1}{2}
\end{align*}
\noindent where $\delta=\mu_2-\mu_1$. 

The posterior distribution is then defined as
\begin{equation*}
\prod_{i=1}^n \left[p_1\mathcal{N}(x_i | \mu_1,\sigma_1)+p_2\mathcal{N}(x_i |\mu_2,\sigma_2)\right]\pi^J(\mu_1,\mu_2|p,\sigma).
\end{equation*}

The likelihood may be rewritten (without loss of generality, by considering $\sigma_1=\sigma_2=1$, since they are known) as
\begin{align}
L(\theta)&=\prod_{i=1}^n \left[p_1\mathcal{N}(x_i |\mu_1,1)+p_2\mathcal{N}(x_i |\mu_2,1)\right] \nonumber \\
		&= \frac{1}{(2\pi)^\frac{n}{2}}\left[p_1^n e^{-\frac{1}{2}\sum\limits_{i=1}^n (x_i-\mu_1)^2}+\sum_{j=1}^n p_1^{n-1}p_2e^{-\frac{1}{2}\sum\limits_{i\neq j} (x_i-\mu_1)^2-\frac{1}{2}(x_j-\mu_2)^2}\right. \nonumber \\
		&{} \left. +\sum_{j=1}^n \sum_{k\neq j} p_1^{n-2}p_2^2 e^{-\frac{1}{2}\left[ \sum\limits_{i\neq j,k} (x_i-\mu_1)^2+(x_j-\mu_2)^2+(x_k-\mu_2)^2 \right] } \right. \nonumber \\
		&{} \left. +\cdots+p_2^n e^{-\frac{1}{2}\sum\limits_{j=1}^n (x_j-\mu_2)^2}\right].
\label{eq:mixlik}
\end{align}
Then, for $|\mu_1|\rightarrow\infty$, $L(\theta)$ tends to the term 
$$p_2^n e^{-\frac{1}{2}\sum\limits_{j=1}^n (x_j-\mu_2)^2}$$
that is constant for $\mu_1$. Therefore we can study the behavior of the posterior distribution for this part of the
likelihood to assess its properness. 

This explains why we want the following integral to converge:
\begin{equation*}
\int_{\mathbb{R}\times\mathbb{R}} p_2^n e^{-\frac{1}{2}\sum\limits_{j=1}^n (x_j-\mu_2)^2} \pi^J(\mu_1,\mu_2) d\mu_1 d\mu_2
\end{equation*}
which is equal to (by the change of variable $\mu_2-\mu_1=\delta$)
\begin{equation*}
\int_{\mathbb{R}\times\mathbb{R}} p_2^n e^{-\frac{1}{2}\sum\limits_{j=1}^n (x_j-\mu_1-\delta)^2} \pi^J(\mu_1,\delta) d\mu_1 d\delta.
\end{equation*}
In Appendix C of the Supplementary Material it is possible to see that the prior distribution only depends on the difference between the means $\delta$:
\begin{align}
&\int_\mathbb{R} p_2^n \int_\mathbb{R} e^{-\frac{1}{2}\sum\limits_{j=1}^n (x_j-\mu_1-\delta)^2}d\mu_1 \pi^J(\delta)d\delta \nonumber \\
&\propto \int_\mathbb{R} \int_\mathbb{R} e^{-\frac{1}{2}\sum\limits_{j=1}^n (x_j-\delta)^2
+\mu_1\sum\limits_{j=1}^n(x_j-\delta)-\frac{1}{2}n\mu_1^2} d\mu_1 \pi^J(\delta)d\delta \nonumber \\
&=\int_\mathbb{R} \left[\int_\mathbb{R} e^{\mu_1\sum\limits_{j=1}^n(x_j-\delta)-\frac{1}{2}n\mu_1^2} d\mu_1\right]
e^{-\frac{1}{2}\sum\limits_{j=1}^n (x_j-\delta)^2} \pi^J(\delta)d\delta \nonumber \\
&=\int_\mathbb{R} e^{-\frac{1}{2}\sum\limits_{j=1}^n (x_j-\delta)^2+\sum\limits_{j=1}^n\frac{(x_j-\delta)}{2n}} \pi^J(\delta)d\delta \approx \int_\mathbb{R} e^{-\frac{n-1}{2}\delta^2} \pi^J(\delta)d\delta. \label{eq:postmean}
\end{align}
where $\pi^J(\delta)$ is defined as 
\begin{equation*}
\pi^J(\delta)\propto \left[
\int_\mathfrak{X}\frac{\left[{(1-p)x\exp\{-\frac{x^2}{2}\}}\right]^2}{{p\sigma\exp\{-\frac{\sigma^2(x+\frac{\delta}{\sigma\tau})^2}{2}\}}+{(1-p)\exp\{-\frac{x^2}{2}\}}}
d x \right]^{\frac{1}{2}}
\end{equation*}
As $\delta \rightarrow \pm \infty$ this quantity is constant with respect to $\delta$. Therefore the integral \eqref{eq:postmean} is convergent for $n \geq 2$.
   
\end{proof} 

\subsection*{Appendix E: Proof of Lemma 3.4}

\textit{
When $k>2$, the posterior distribution derived from the Jeffreys prior when only the location parameters of model \eqref{eq:mix-locscale} are unknown is improper.
}

\begin{proof}
In the case of $k\neq 2$ components, the Jeffreys prior for the location parameters is still constant with respect to a
reference mean (for example, $\mu_1$).  Therefore it depends on the difference parameters
$(\delta_2=\mu_2-\mu_1,\delta_3=\mu_3-\mu_1,\cdots,\delta_k=\mu_k-\mu_1)$.

The Jeffreys prior depends on the product on the diagonal

\begin{align*}
& \bigintsss_{-\infty}^\infty \frac{[t\mathcal{N}(0,\sigma_1^2)]^2}{p_1\mathcal{N}(0,\sigma_1^2)+\cdots+p_k \mathcal{N}(\delta_k,\sigma_k^2)}d t  \nonumber \\
							& \cdots \bigintsss_{-\infty}^\infty \frac{[u \mathcal{N}(0,\sigma_k^2)]^2}{p_1\mathcal{N}(-\delta_k,\sigma_1^2)+\cdots+p_k \mathcal{N}(0,\sigma_k^2)} d u .
\end{align*}

If we consider the case as in Lemma 3.3, where only the part of the likelihood depending on e.g. $\mu_2$
may be considered, the convergence of the following integral has to be studied:

\begin{equation*}
\int_\mathbb{R} \cdots \int_\mathbb{R} e^{-\frac{n-1}{2}\delta_2^2} \pi^J(\delta_2,\cdots,\delta_k) d \delta_2 \cdots d \delta_k.
\end{equation*}

In this case, however, the integral with respect to $\delta_2$ may converge, nevertheless the integrals with respect to $\delta_j$ with $j\neq 2$ will diverge, since the prior tends to be constant for each $\delta_j$ as $|\delta_j| \rightarrow \infty$.
\end{proof}

\subsection*{Appendix F: Proof of Theorem 3.1}

\textit{The posterior distribution of the parameters of a mixture model with location-scale components derived from the Jeffreys prior when all parameters of model \eqref{eq:mix-locscale} are unknown is improper.}

\begin{proof}

Consider a mixture model with components coming from the location-scale family. The proof will consider Gaussian components, however it may be generalized to any location-scale distribution. 

Consider the elements on the diagonal of the Fisher information matrix; again, since the Fisher information matrix is positive definite, the determinant is bounded by the product of the terms in the diagonal. 

Consider a reparametrization into $\tau=\sigma_1$ and $\tau\sigma=\sigma_2$. Then it is straightforward to see that the integral of this part of the prior distribution will depend on a term $(\tau)^{-(d+1)}(\sigma)^{-d}$, as seen in the proof of Lemma 3.2. 
The likelihood, on the other hand, is given by

\begin{align}
\label{eq:mixlik}
L(\theta)&=\prod_{i=1}^n \left[p\mathcal{N}(\mu,\tau^2)+(1-p)\mathcal{N}(\mu + \tau \delta,\tau^2\sigma^2)\right] \nonumber \\
		&= \frac{1}{(2\pi)^\frac{n}{2}}\left[\frac{1}{\tau^n} p^n e^{-\frac{\sum\limits_{i=1}^n (x_i-\mu)^2}{2\tau^2}}+ \right. \nonumber \\
		&{} \left.  + \frac{1}{\tau^{n}\sigma}\sum\limits_{i=1}^n p^{n-1}(1-p)e^{-\frac{\sum\limits_{j\neq i} (x_j-\mu_1)^2}{2\tau^2}-\frac{(x_i-\mu_2)^2}{2\tau^2\sigma^2}}\right. \nonumber \\
		&{} \left. +\frac{1}{\tau^n\sigma^2}\sum_{i=1}^n \sum_{k\neq i} p^{n-2}(1-p)^2 e^{-\frac{\sum\limits_{j\neq i,k} (x_j-\mu)^2}{2\tau^2}} \cdot \right. \nonumber \\
		&{} \left.  \cdot e^{-\frac{\left[(x_i-(\mu+\tau\delta))^2+(x_k-(\mu+\tau\delta))^2 \right]}{2\tau^2\sigma^2}} \right. \nonumber \\
		&{} \left. +\cdots+(1-p)^n \frac{1}{\tau^2\sigma^2} e^{-\frac{1}{2}\sum\limits_{i=1}^n (x_i-(\mu+\tau\delta))^2}\right].
\end{align}

When composing the prior with the part of the likelihood which only depends on the
first component, this part does not provide information about the parameters $\sigma$ and the integral will diverge. 

In particular, the integral of the first part of the posterior distribution relative to the part of the likelihood dependent on the first component only and on the product of the diagonal terms of the Fisher information matrix for the prior when considering a two-component mixture model is

\footnotesize{
\begin{align*}
\int_0^1  & \int_{-\infty}^{+\infty} \int_{-\infty}^{+\infty} \int_0^{\infty} \int_0^{\infty} c \frac{p_1^n}{\tau^n}\frac{p_1^2 p_2^2}{\tau^3\sigma^2} \exp \left\{-\frac{1}{2\tau^2}\sum_{i=1}^n(x_i-\mu_1)^2\right\} \nonumber \\
				& \times{} \left\{ \int_{-\infty}^{\infty} \frac{\left[ \sigma\exp\left\{-\frac{(\tau\sigma y + \delta)^2}{2\tau^2} \right\} - \exp\left\{-\frac{y^2}{2} \right\}\right]^2}{p_1 \sigma \exp\left\{-\frac{(\tau \sigma y + \delta)^2}{2\tau^2}\right\}+p_2\exp\left\{-\frac{y^2}{2}\right\}}
d y \right. \nonumber \\
				& \times \left.{} \int_{-\infty}^{\infty} \frac{z^2 \exp(-z^2)}{p_1 \exp\left\{-\frac{z^2}{2}\right\}+\frac{p_2}{\sigma}\exp\left\{-\frac{(z\tau-\delta)^2}{2\tau^2\sigma^2}\right\}}
d z \right. \nonumber \\
				& \times \left.{} \int_{-\infty}^\infty
\frac{w^2 \exp\left\{ -w^2\right\}}{p_1 \sigma \exp\left\{-\frac{(\tau \sigma w+\delta)^2}{2\tau^2\sigma^2}\right\}+p_2\exp\left\{-\frac{w^2}{2}\right\}}
d w \right. \nonumber \\
				& \times \left.{} \int_{-\infty}^\infty
\frac{\left(z^2-1\right)^2 \exp\left\{ -z^2\right\}}{p_1 \exp\left\{-\frac{z^2}{2}\right\}+\frac{p_2}{\sigma}\exp\left\{-\frac{(z\tau+\mu_1-\mu_2)^2}{2\tau^2\sigma^2}\right\}}
d z \right. \nonumber \\
				&\times \left.{}  \int_{-\infty}^\infty
\frac{\left(u^2-1\right)^2 \exp\left\{-u^2\right\}}{p_1\sigma \exp\left\{-\frac{(u\tau\sigma+\mu_2-\mu_1)^2}{2\tau^2}\right\}+p_2\exp\left\{-\frac{u^2}{2}\right\}} 
d u \right\}^\frac{1}{2} \nonumber  \\
& \qquad d \tau  d \sigma d \mu_1 d \mu_2 d p_1. \nonumber
\end{align*}
}

\normalsize{
When considering the integrals relative to the Jeffreys prior, they do not
represent an issue for convergence with respect to the scale parameters,
because exponential terms going to $0$ as the scale parameters tend to $0$ are
present. However, when considering the part out of the previous integrals, a
factor $\sigma^{-2}$ which causes divergence is present. Then this particular
part of the posterior distribution does not integrate.

When considering the case of $k$ components, the integral inversily
depends on $\sigma_1, \sigma_2,\cdots, \sigma_{k-1}$ which implies the posterior
always is improper. 
}
\end{proof}

\subsection*{Appendix G: Improperness of the posterior distribution deriving from the multivariate Jeffreys prior}

Since the posterior distribution which follows from the use of the multivariate Jeffreys prior on the complete set of parameters is improper, we expect to see non-convergent behaviors in the MCMC simulations, in particular for small sample sizes. 
For small sample sizes, the chains tend to get stuck when very small values of standard deviations are accepted. Figure \ref{fig:improper_post} shows the results for different sample sizes and different scenarios (in particular, the situations when the means are close or well separated from one another are considered) for a mixture model with two and
three Gaussian components: sometimes the chains do not converge and tend
towards very extreme values of means, sometimes the chains get stuck to very small
values of standard deviations. 

\begin{figure}
\centering
\includegraphics[scale=0.5]{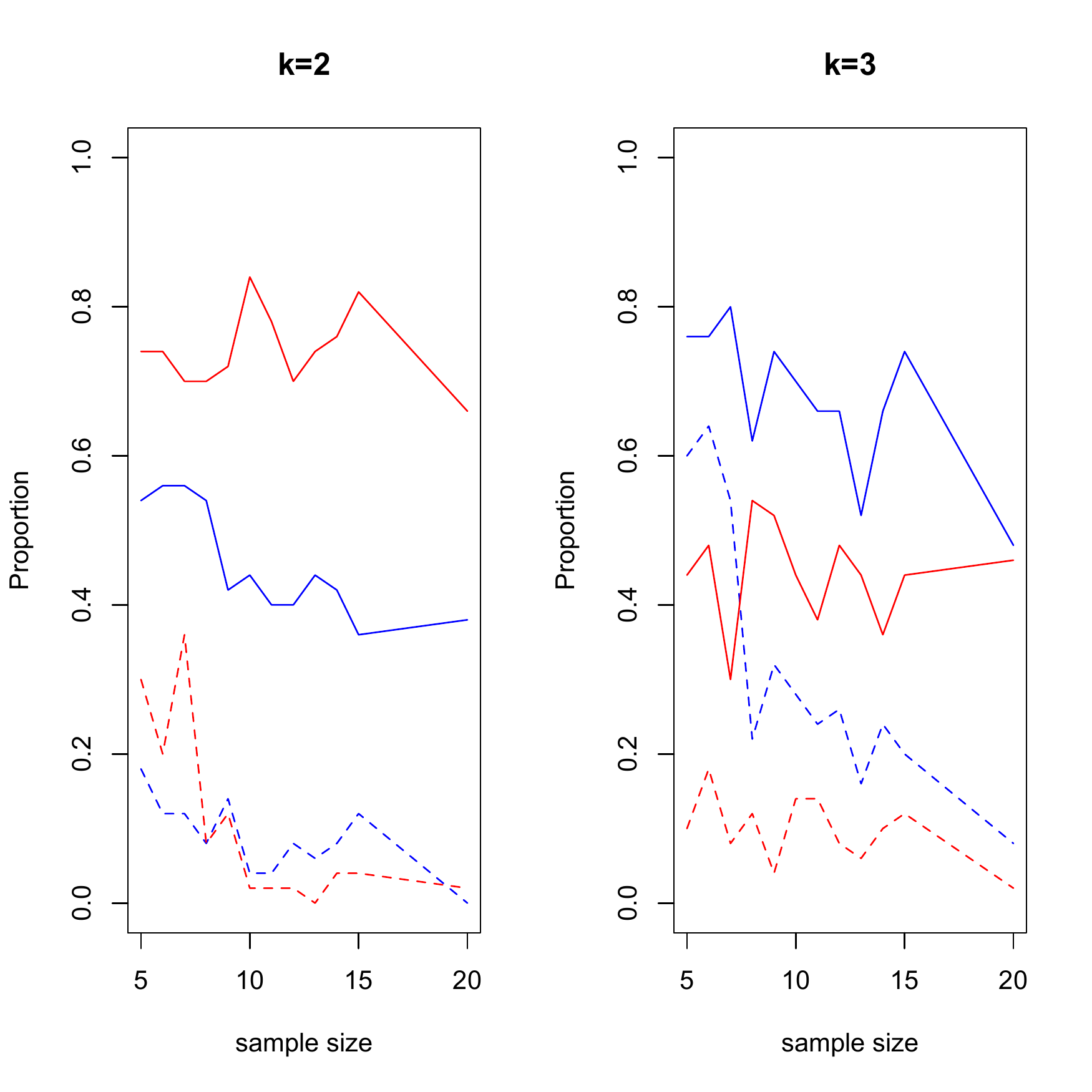}
\caption{All parameters unknown: results from 50 replications of the experiment with close means (solid lines) and well-separated means (dashed lines) based on $10^5$ simulations and a burn-in of $10^4$ simulations. The graph shows the proportion of Monte Carlo chains stuck at values of standard deviations close to zero (blue lines) and the proportion of chains diverging towards high values of means. The case of a two-component GMM is on the left, the case of a three-component GMM is on the right.}
\label{fig:improper_post}
\end{figure}

The improperness of the posterior distribution is not only due to the scale
parameters: we may use a reparametrization of the problem as in Equation
(8) and use a proper prior on the parameter $\sigma$, for example,
by following \cite{robert:mengersen:1999}

\begin{equation*}
	p(\sigma)=\frac{1}{2}\mathcal{U}_{[0,1]}(\sigma)+\frac{1}{2}\frac{1}{\mathcal{U}_{[0,1]}(\sigma)}.
\end{equation*}

\noindent and the Jeffreys prior for all the other parameters
$(p,\mu,\delta,\tau)$ conditionally on $\sigma$, and still face the
same issue. Actually, using a proper prior on $\sigma$ does not avoid
convergence trouble, as demonstrated by Figure \ref{fig:improper_post_props}, which shows
that, even if the chains with respect to the standard deviations are not stuck
around $0$ when using a proper prior for $\sigma$ in the reparametrization
proposed by \cite{robert:mengersen:1999}, the chains with respect to the
locations parameters demonstrate a divergent behavior.

\begin{figure}
\centering
\includegraphics[scale=0.5]{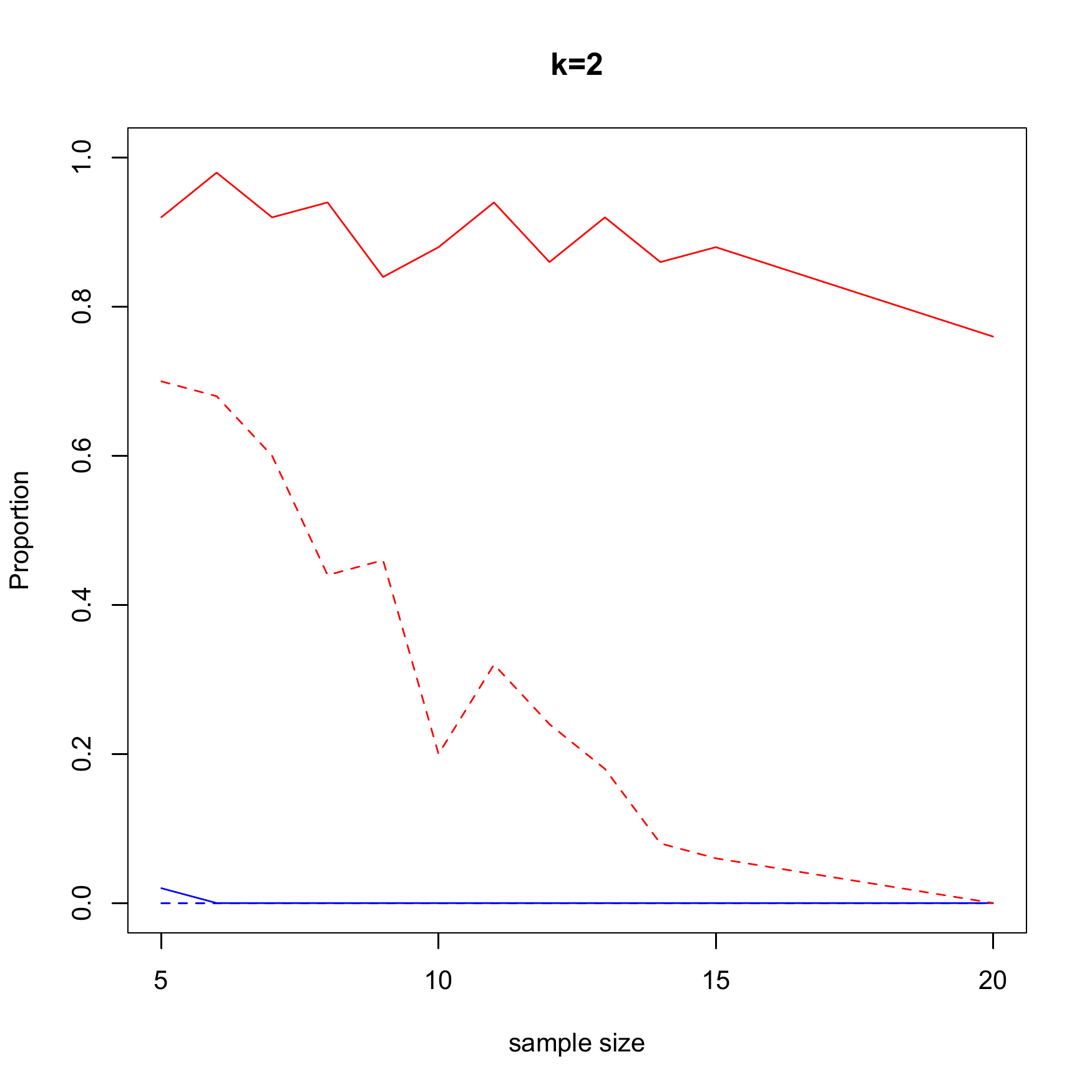}
\caption{All parameters unknown, proper prior on $\sigma$, two-component GMM: results from 50 replications of the experiment for both close means (solid lines) and far means (dashed lines) based on $10^5$ simulations and a burn-in of $10^4$ simulations. The graph shows the proportion of Monte Carlo chains stuck at values of standard deviations close to 0 (blue lines) and the proportion of chains diverging towards high values of means (red lines).}
\label{fig:improper_post_props}
\end{figure}

These problems are overcome by the hierarchical prior proposed in Section 4: a simulation study (not shown) along the lines of the one just presented for the posterior distribution deriving from the multivariate Jeffreys prior confirms that the chains obtained via MCMC for 50 replications of the experiments always have a convergent behavior despite the posterior being improper. 

\subsection*{Appendix H: The properness of the hierarchical representation of Theorem 4.1}

\textit{The posterior distribution derived from the hierarchical representation of the Gaussian mixture model 
associated with (9), (10) and (11)
is proper. }

\begin{proof}
Consider the composition of the three levels of the hierarchical model described in equations (9), (10) and (11):

\begin{align}
\label{eq:hierarch_post}
\pi(\boldsymbol{\mu},\boldsymbol{\sigma},&\mu_0,\zeta_0;\mathbf{x})  \propto L(\mu_1,\mu_2,\sigma_1,\sigma_2;\mathbf{x})  p^{-1/2} (1-p)^{-1/2}
			\nonumber \\
			& {} \times \frac{1}{\zeta_0} \frac{1}{2\pi\zeta_0^2} \exp\left\{- \frac{(\mu_1-\mu_0)^2 (\mu_2-\mu_0)^2}{2\zeta_0^2}\right\} \nonumber \\ 
			& {} \times \left[ \frac{1}{2}\frac{1}{\zeta_0} \mathbb{I}_{[\sigma_1\in(0,\zeta_0)]}(\sigma_1) + \frac{1}{2}\frac{\zeta_0}{\sigma_1^2} \mathbb{I}_{[\sigma_1\in(\zeta_0,+\infty)]}(\sigma_1) \right] \nonumber \\
			& {} \times \left[ \frac{1}{2}\frac{1}{\zeta_0} \mathbb{I}_{[\sigma_2\in(0,\zeta_0)]}(\sigma_2) + \frac{1}{2}\frac{\zeta_0}{\sigma_2^2} \mathbb{I}_{[\sigma_2\in(\zeta_0,+\infty)]}(\sigma_2) \right]
\end{align}

\noindent where $L(\cdot;\mathbf{x})$ is given by Equation \eqref{eq:mixlik}. 

Once again, we can initialize the proof by considering only the first term in the sum composing the likelihood function
for the mixture model. Then the product in \eqref{eq:hierarch_post} may be split into four terms corresponding to the
different terms in the scale parameters' prior. For instance, the first term is

\begin{align*}
\int_0^\infty & \int_{-\infty}^\infty \int_\mathbb{R}\int_\mathbb{R} \int_\mathbb{R^+} \int_\mathbb{R^+} \int_0^1 
				\frac{1}{\sigma_1^n} p_1^n \exp \left\{- \frac{\sum\limits_{i=1}^n (x_i-\mu_1)^2}{2\sigma_1^2} \right\} \nonumber \\
				& {} \times \frac{1}{\zeta_0^3} \exp\left\{-\frac{(\mu_1-\mu_0)^2 (\mu_2-\mu_0)^2}{2\zeta_0^2} \right\} \nonumber \\
				& {} \times \frac{1}{4}\frac{1}{\zeta_0} \frac{1}{\zeta_0} \mathbb{I}_{[\sigma_1 \in (0,\zeta_0)]}(\sigma_1) \mathbb{I}_{[\sigma_2 \in (0,\zeta_0)]}(\sigma_2) \nonumber \\
& d p d\sigma_1 d\sigma_2 d\mu_1 d\mu_2 d \mu_0 d\zeta_0 
\end{align*}

\noindent and the second one

\begin{align*}
\int_0^\infty & \int_{-\infty}^\infty \int_\mathbb{R}\int_\mathbb{R} \int_\mathbb{R^+} \int_\mathbb{R^+} \int_0^1 
				\frac{1}{\sigma_1^n} p_1^n \exp \left\{- \frac{\sum\limits_{i=1}^n (x_i-\mu_1)^2}{2\sigma_1^2} \right\} \nonumber \\
				& {} \times \frac{1}{\zeta_0^3} \exp\left\{-\frac{(\mu_1-\mu_0)^2 (\mu_2-\mu_0)^2}{2\zeta_0^2} \right\} \nonumber \\
				& {} \times \frac{1}{4}\frac{1}{\zeta_0} \frac{\zeta_0}{\sigma_2^2} \mathbb{I}_{[\sigma_1 \in (0,\zeta_0)]}(\sigma_1) \mathbb{I}_{[\sigma_2 \in (\zeta_0,\infty)]}(\sigma_2) \nonumber \\
& d p d\sigma_1 d\sigma_2 d\mu_1 d\mu_2 d \mu_0 d\zeta_0 .
\end{align*}

The integrals with respect to $\mu_1$, $\mu_2$ and $\mu_0$ converge, since the data are carrying
information about $\mu_0$ through $\mu_1$. The integral with respect to $\sigma_1$ converges as well, because, as
$\sigma_1 \rightarrow 0$, the exponential function goes to $0$ faster than $\frac{1}{\sigma_1^n}$ goes to $\infty$
(integrals where $\sigma_1>\zeta_0$ are not considered here because this reasoning may easily extend to those
cases).  The integrals with respect to $\sigma_2$ 
converge, because they provide a factor proportional to $\zeta_0$ and $1/\zeta_0$ respectively
which simplifies with the normalizing constant of the reference distribution (the uniform in the first case and the
Pareto in second one). Finally, the term $1/\zeta_0^4$ resulting from the previous operations has its counterpart in the
integrals relative to the location priors. Therefore, the integral with respect to $\zeta_0$ converges. 

The part of the posterior distribution relative to the weights is not an issue, since the weights belong to the
corresponding simplex.
\end{proof}

\subsection*{Appendix I: Effect of the sample size in the conservativeness of the Jeffreys prior}

This Appendix shows the estimation of the density \eqref{eq:overfit} when a higher number of components is assumed, together with a Jeffreys prior for the weigths of the mixture for sample sizes $50, \, 100, \, 500, \, 1,000$, when the true model is 

\begin{equation}
\label{eq:overfit}
0.5 \mathcal{N}(-3,1) + 0.5 \mathcal{N}(3,1).
\end{equation}

Figures \ref{fig:hierch-overfitting-dens50}-\ref{fig:hierch-overfitting-dens1000} show
the $M=20$ resulting estimated densities against \eqref{eq:overfit}; as the
number of components increases, the estimated densities are less and less
smooth, nevertheless this feature is mitigated as the sample size increases.

\begin{figure}
\centering
\includegraphics[scale=0.5]{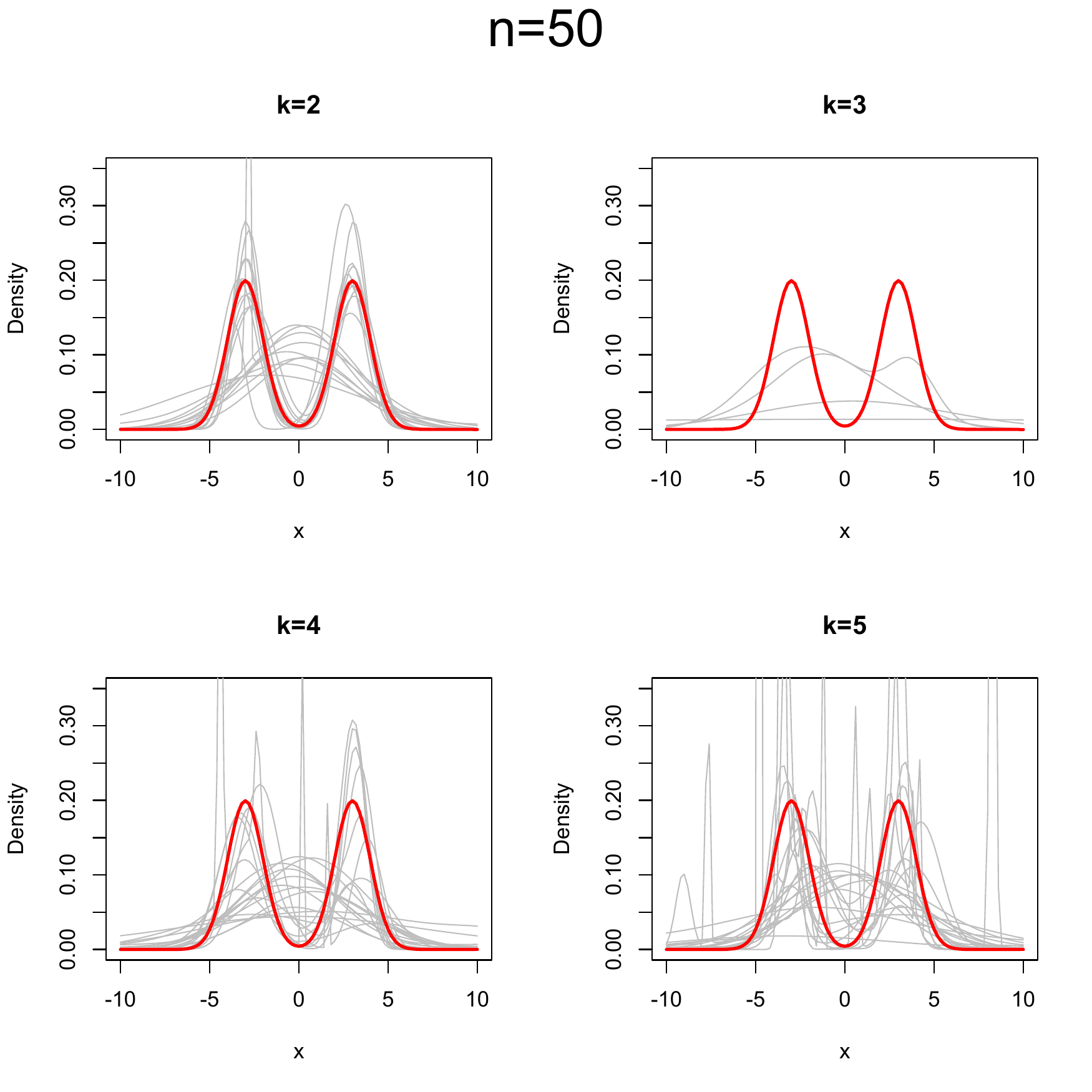}
\caption{Estimated densities in $20$ replications of the experiment (in grey) against the true model (in red) for $n=50$.}
\label{fig:hierch-overfitting-dens50}
\end{figure}

\begin{figure}
\centering
\includegraphics[scale=0.5]{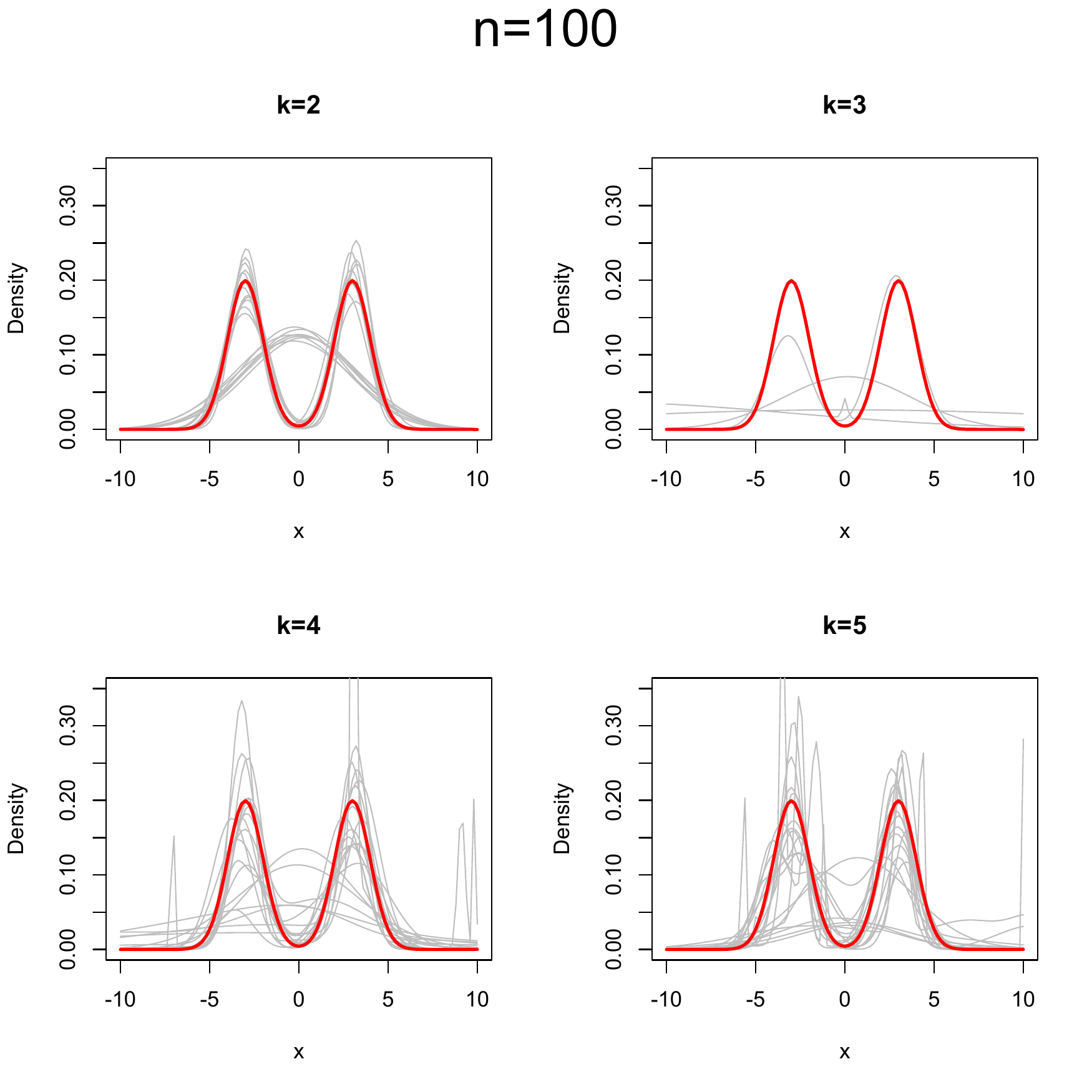}
\caption{As in Figure \ref{fig:hierch-overfitting-dens50}, for $n=100$.}
\label{fig:hierch-overfitting-dens100}
\end{figure}

\begin{figure}
\centering
\includegraphics[scale=0.5]{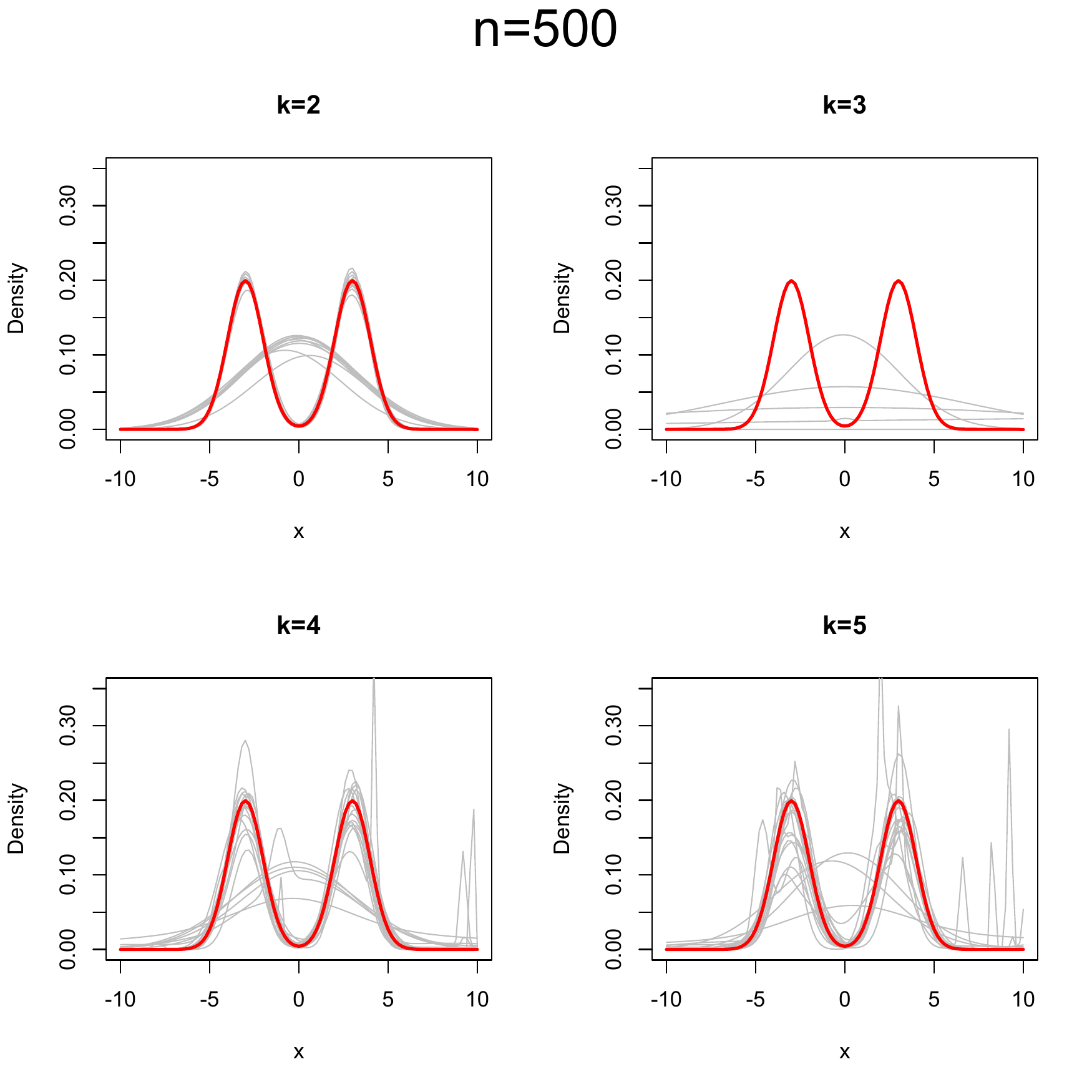}
\caption{As in Figure \ref{fig:hierch-overfitting-dens50}, for $n=500$.}
\label{fig:hierch-overfitting-dens500}
\end{figure}

\begin{figure}
\centering
\includegraphics[scale=0.5]{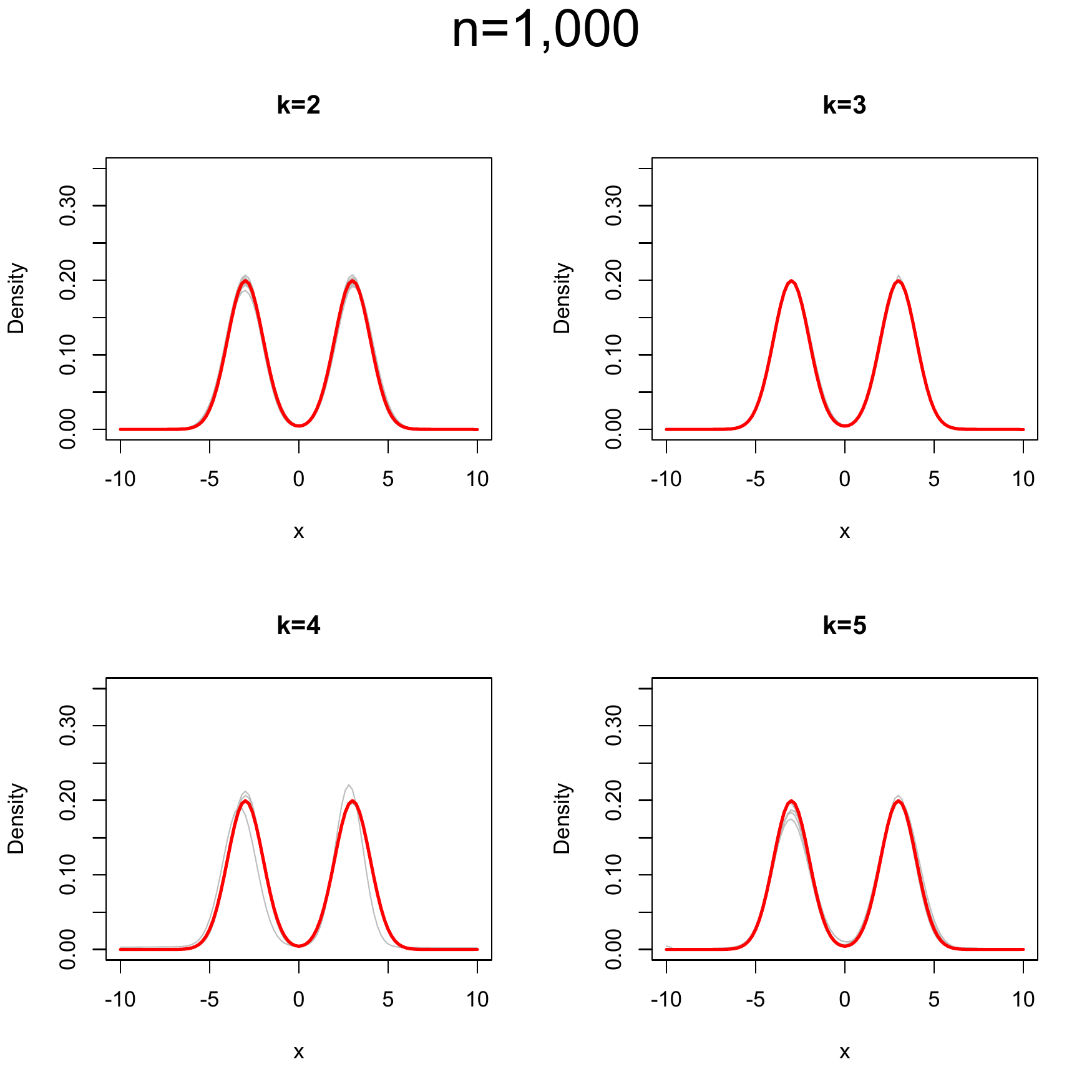}
\caption{As in Figure \ref{fig:hierch-overfitting-dens50}, for $n=1000$.}
\label{fig:hierch-overfitting-dens1000}
\end{figure}

\subsection*{Appendix J: Implementation Features}

The computing expense due to derive the Jeffreys prior for a set of parameter values is in $\mathrm{O}(d^2)$ if
$d$ is the total number of (independent) parameters.

Each element of the Fisher information matrix is an integral of the form

\begin{equation*}
-\int_{\mathcal{X}} \frac{\partial^2 \log \left[\sum_{h=1}^k p_h\,f(x|\theta_h)\right]}{\partial \theta_i \partial \theta_j}\left[\sum_{h=1}^k p_h\,f(x|\theta_h)\right]^{-1} d x
\end{equation*}

\noindent which has to be approximated. We have applied both numerical
integration and Monte Carlo integration and simulations show that, in general,
numerical integration obtained via Gauss-Kronrod quadrature, produces more stable results. Nevertheless, when the values of one or more 
standard deviations or weights are too
small, either the approximations tend to be very dependent on the bounds used for
numerical integration (usually chosen to omit a negligible part of the
density) or the numerical approximation may not be even applicable. 
In this case, Monte Carlo integration seems to be more stable, where the stability of the results depends on the Monte Carlo sample size. 

Figure \ref{fig:MCvsNUM_incrN} shows the value of the Jeffreys prior obtained via Monte Carlo integration of the
elements of the Fisher information matrix for an increasing number of Monte Carlo simulations both in the case where the
Jeffreys prior diverges (where the standard deviations are small) and  where it assumes low values. The value
obtained via Monte Carlo integration is then compared with the value obtained via numerical integration. The sample size
relative to the point where the graph stabilizes may be chosen to perform the approximation. The number of Monte Carlo simulations needed to reach a fixed amount of variability may be chosen. 

\begin{figure}
\centering
\includegraphics[scale=0.4]{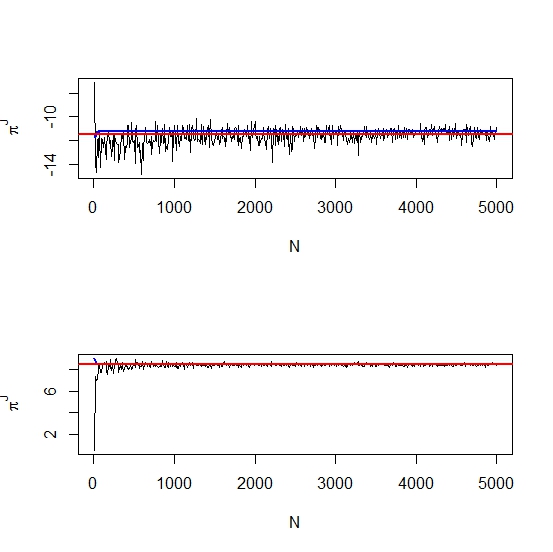}
\caption{Jeffreys prior obtained via Monte Carlo integration (and numerical integration, in \textit{red}) for the model $0.25\mathcal{N}(-10,1)+0.10\mathcal{N}(0,5)+0.65\mathcal{N}(15,7)$ (above) and for the model $\frac{1}{3}\mathcal{N}(-1,0.2)+\frac{1}{3}\mathcal{N}(0,0.2)+\frac{1}{3}\mathcal{N}(1,0.2)$ (below). On the $x$-axis there is the number of Monte Carlo simulations.}
\label{fig:MCvsNUM_incrN}
\end{figure}

Since the approximation problem is one-dimensional, another numerical solution could be based on the Riemann sums;
Figure \ref{fig:MCvsRIEMbxp} shows the comparison between the approximation to
the Jeffreys prior obtained via Monte Carlo integration and via the sums of Riemann: it is clear that the
Riemann sums lead to more stable results in comparison with Monte Carlo integration. On the other hand, they can be applied
in more situations than the Gauss-Kromrod quadrature, in particular, in cases where the standard deviations are very
small (of order $10^{-2}$). Nevertheless, when the standard deviations are smaller than this, one has to pay attention
on the features of the function to integrate. In fact, the mixture density tends to concentrate around the modes, with
regions of density close to 0 between them. The elements of the Fisher informtation matrix are, in general, ratios
between the components' densities and the mixture density, then in those regions an indeterminate form of type
$\frac{0}{0}$ is obtained; Figure \ref{fig:FishInfoelem} represents the behavior of one of these elements when $\sigma_i
\rightarrow 0$ for $i=1,\cdots,k$. 

\begin{figure}
\centering
\includegraphics[scale=0.4]{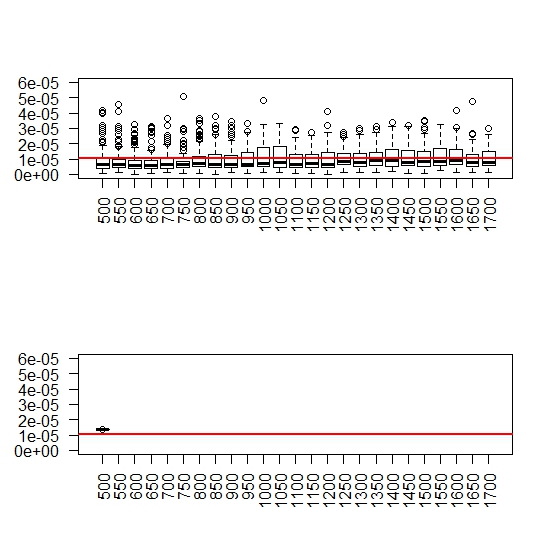}
\caption{Boxplots of 100 replications of the procedure based on Monte Carlo integration (above) and Riemann sums (below) which approximates the Fisher information matrix of the model $0.25\mathcal{N}(-10,1)+0.10\mathcal{N}(0,5)+0.65\mathcal{N}(15,7)$ for sample sizes from $500$ to $1700$. The value obtained via numerical integration is represented by the red line (in the graph below, all the approximations obtained with more than $550$ knots give the same result, exactly equal to the one obtained via Riemann sums).}
\label{fig:MCvsRIEMbxp}
\end{figure}

\begin{figure}
\centering
\includegraphics[scale=0.4]{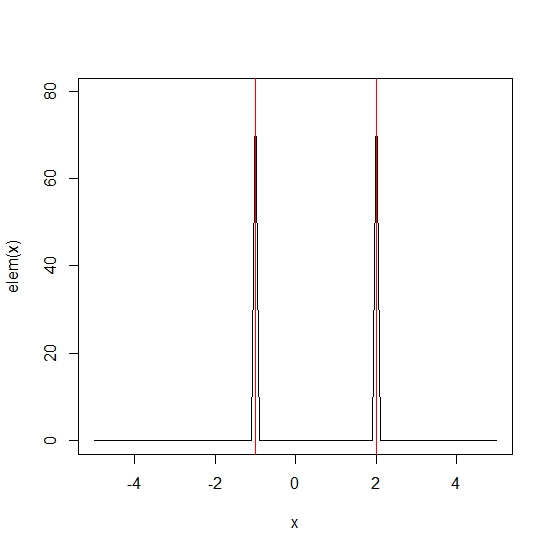}
\caption{The first element on the diagonal of the Fisher information matrix relative to the first weight of the two-component Gaussian mixture model $0.5 \mathcal{N}(-1,0.01)+0.5 \mathcal{N}(2,0.01)$.}
\label{fig:FishInfoelem}
\end{figure}

Thus, we have decided to use the Riemann sums (with a number of points equal to $550$) to approximate the Jeffreys prior when the standard deviations are sufficiently large and Monte Carlo integration (with sample sizes of $1500$) when they are too small. In this case, the variability of the results seems to decrease as $\sigma_i$ approaches $0$, as shown in Figure \ref{fig:MCsmallsd}.
 
\begin{figure}
\centering
\includegraphics[scale=0.4]{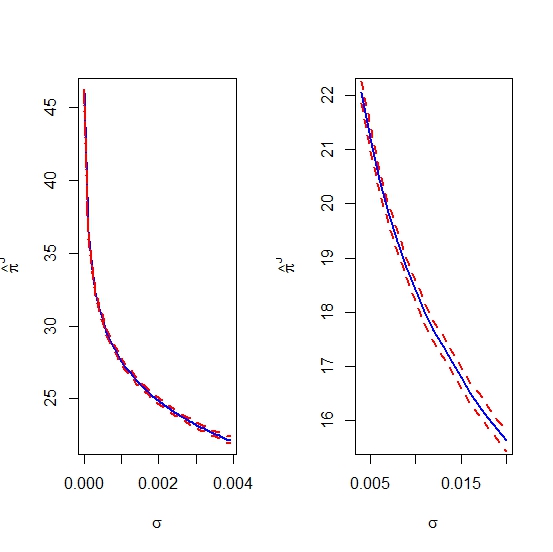}
\caption{Approximation of the Jeffreys prior (in log-scale) for the two-component Gaussian mixture model $0.5 \mathcal{N}(-1,\sigma)+0.5\mathcal{N}(2,\sigma)$, where $\sigma$ is taken equal for both components and increasing.}
\label{fig:MCsmallsd}
\end{figure}

We have chosen to consider Monte Carlo samples of size equal to $1500$ because both the value of the approximation and
its standard deviations are stabilizing.

An adaptive MCMC algorithm has been used to define the variability of the
kernel density functions used to propose the moves. During the burnin, the
variability of the kernel distributions has been reduced or increased depending
on the acceptance rate, in a way such that the acceptance rate stay between $20\%$
and $40\%$. The transitional kernel used have been truncated normals for the
weights, normals for the means and log-normals for the standard deviations (all
centered on the values accepted in the previous iteration).

\end{document}